\documentclass[]{article}

\usepackage{amsmath}
\usepackage{amssymb}
\usepackage{amsfonts}
\usepackage{epstopdf}
\usepackage{url}
\usepackage{multirow}
\usepackage{tablefootnote}
\usepackage{dsfont}
\usepackage{algorithm}
\usepackage{algorithmic}
\usepackage{subcaption}
\usepackage{amsthm}
\usepackage{graphicx}
\newtheorem{rmk}{Remark}[section]
\newtheorem{pro}{Proposition}[section]
\newtheorem{coro}{Corollary}[section]

\newcommand{\na}{a_\nu}
\newcommand{\nb}{b_\nu}
\newcommand{\ra}{a_\rho}
\newcommand{\rb}{b_\rho}
\newcommand{\sa}{a_s}

\newcommand{\xf}{\widetilde{f}}

\newcommand*{\bW}[2]{W_{#1}^{#2}}
\newcommand*{\tbm}[2]{\tilde{W}_{#1}^{#2}}
\newcommand{\artanh}{\text{artanh}}
\newcommand{\deriv}[3][]{
	\ensuremath{\frac{\partial^{#1} {#2}}{\partial {#3}^{#1}}}
}
\newcommand{\norm}[1]{\left\lVert#1\right\rVert}
\usepackage{natbib}
\bibpunct[, ]{(}{)}{;}{a}{}{,}

\theoremstyle{plain}
\newtheorem{theorem}{Theorem}[section]

\theoremstyle{definition}

\theoremstyle{remark}

\begin{document}

\title{On the Pricing of American Options under Stochastic Local Volatility and Stochastic Correlation via the RBSDE Framework}

\author{
	L. Teng\thanks{Email: teng@math.uni-wuppertal.de}\\
	Applied and Computational Mathematics\\
	University of Wuppertal\\
	Wuppertal, Germany
}

\maketitle

\begin{abstract}
In this work, we study the pricing of American options under stochastic local volatility (SLV) models extended by including stochastic correlation driven by an additional stochastic process. We generalize the class of SLV models by incorporating a flexible stochastic correlation structure. To price options within these extended models, we derive the corresponding reflected forward-backward stochastic differential equations (RBSDEs) and employ data-driven numerical methods to solve them for both pricing and hedging purposes. The RBSDE framework enables the modelling of the future evolution of the option price. Furthermore, we conduct a convergence analysis of the proposed numerical method and present numerical experiments that illustrate the performance of the extended models, as well as the accuracy and efficiency of the RBSDE-based approach.
\end{abstract}


\section{Introduction}

Local volatility (LV) models were proposed in \cite{Dupire1994,Derman1994} for an exact and fast fit to the implied volatility surface of European option prices with the Dupire's formula. However, it has been pointed out in \cite{Hagan2002} that these self-consistent and arbitrage-free models predict the wrong dynamics of the implied volatility curve, leading to unstable hedges of exotic options. An alternative class of parametric models are stochastic volatility (SV) models, e.g. \cite{Heston1993, Hull1987}, which can provide a more realistic term structure for implied volatility surface, in particular future smiles. Stochastic local volatility (SLV) models can be constructed by combing desirable features from both LV models and SV models, the well-known example is the stochastic Alpha Beta Rho (SABR) model \cite{Hagan2002}.

Despite those features of the SLV models, it is quite challenging to compute option prices and hedging quantities due to the correlated Brownian motions and the nonlinear local volatility structure in the price process. In the literature so far, several Monte Carlo simulation methods have been proposed; see, e.g., \cite{Cai2017, Chen2012, Leitao2017} and references therein. Another class of methods is to use asymptotic expansion techniques \cite{Antonov2013, Armstrong2017, Hagan2002} and many others. The performance of the pricing using these methods is not consistent while varying all parameters, e.g., it is well known that the expansion of implied volatility in \cite{Hagan2002} is less accurate for a large maturity, or when the options are deep in or out of the money. This is one of the reasons why a double layer continuous-time Markov chain (CTMC) approximation method is proposed in \cite{Cui2018}, in which the pricing of American options is also considered. This CTMC approximation method is more general and applicable for many SLV models with an arbitrary correlation structure. In \cite{Cui2017}, a transform-based method is proposed to price swaps and options related to discretely sampled realized variance under a general class of stochastic volatility models with jumps, where a CTMC weak approximation scheme is utilized for the variance process.

To the best of our knowledge, SLV models in the existing literature typically assume constant correlations. However, there is extensive empirical evidence \cite{Driessen2009, Teng2016a, Teng2016b, Teng2016c, Teng2018}, and references therein, including that correlations between financial variables - such as between asset returns or between asset prices and their volatilities - are inherently stochastic. Clearly, assuming constant correlations may lead to misestimation and introduce correlation risk.

In this paper, we propose an extension of standard SLV models extended by incorporating a flexible stochastic correlation structure, resulting in what we refer to as stochastic correlation local volatility (SCLV) models. For pricing options European-style options under the SCLV framework - which leads to a three-dimensional problem even for a single underlying asset - methods such as Monte Carlo simulation (in combination with backward dynamic programming) or Fourier transform techniques may be applicable. However, these methods are not well-suited for pricing American-style options, due to the early exercise feature and the increased model complexity. For instance, the commonly used double-layer CTMC approximation would need to be extended to a triple-layer structure to accommodate the stochastic correlation, which poses significant computational challenges. Moreover, deriving hedge parameters with respect to the local volatility, stochastic correlation, and asset price is generally non-trivial within these frameworks.

Therefore, to address these challenges, we employ an RBSDE-based approach. This framework provides a natural and flexible methodology for handling early-exercise features and the multidimensional structure of the SCLV model. Importantly, the RBSDE formulation in our approach not only facilitates the computation of option prices and their sensitivities, but also enables modelling of the future evolution of option prices in a dynamic and consistent manner.

Research on the pricing of high-dimensional American options remains an active area in the literature; see, for example, \cite{Lelong2018,Redmann2023,Yang2024a}. In particular, several state-of-the-art deep learning-based approaches have been proposed. A deep learning method for optimal stopping problems was first introduced in the pioneering work \cite{Becker2019}, while approaches based on deep neural networks (DNNs) for computing stopping boundaries were first studied in \cite{Reppen2025a}, see also \cite{Reppen2025b}. Many other related contributions, such as \cite{Bayraktar2024,Becker2021a,Dong2024,Hure2020,Reppen2023,Yang2024b}, can be employed for pricing American options under the SCLV models. The deep primal-dual algorithm introduced in the recent work \cite{Yang2024b} provides both lower and upper bounds for American option prices and is therefore naturally chosen as the benchmark for our approach.

We summarize the main contributions of this work as follows:\\
1. We extend classical SLV models by introducing a flexible stochastic correlation structure, governed by a suitable stochastic process, resulting in the SCLV model. This generalization enables effective hedging of correlation risk. Within this framework, we derive the RBSDE pricing and hedging formulas (with respect to the asset price, local volatility, and correlation) for American options.\\
2. The proposed RBSDE framework serves as a general and flexible valuation methodology for SCLV models. The proposed RBSDE algorithm achieves high numerical accuracy, as demonstrated through comparisons with the state-of-the-art methods. In addition, we provide a rigorous convergence analysis to support its reliability. Beyond valuation, the RBSDE framework also enables the modelling of the future evolution of the option prices and hedging strategies, making it a powerful tool for both pricing and risk management.\\
3. To the best of our knowledge, pricing and delta hedging of American options under combined stochastic local volatility and stochastic correlation have not been previously studied. Our results therefore establish important benchmarks for future research in this area.

In the next section, we start with extension of the SLV models by including a stochastic correlation structure, aiming to develop the SCLV models. In Section \ref{sec:RBSDEs} we introduce the RBSDEs for pricing American-type options in the SCLV models. Section \ref{sec:scheme} is devoted to introduction to the XGBoost regression-based scheme including error analysis. The applications of the developed RBSDE pricing framework to the classical SABR and Heston stochastic local volatility (HSLV) models extended with stochastic correlation structure are addressed in Section \ref{sec:application}.
In Section \ref{sec:num}, we provide a set of numerical experiments for American options to illustrate the accuracy and efficiency of the proposed RBSDE approach. Finally, Section \ref{sec:conclusion} concludes this work.

\section{Extension of general stochastic local volatility model with stochastic correlation}\label{sec:ExtendedModel}
Throughout the paper, we assume that $(\Omega, \mathcal{F}, P; \{\mathcal{F}_t\}_{0\leq t\leq T})$ is a complete, filtered probability space.
In this space, the standard $d$-dimensional Brownian motions $W_t$ with a finite terminal time $T$ are defined, which generates the filtration $\{\mathcal{F}_t\}_{0\leq t\leq T},$ i.e., $\mathcal{F}_t=\sigma\{W_s, 0\leq s \leq t\}.$ And the usual hypotheses should be satisfied. We denote the set of all $\mathcal{F}_t$-adapted and square
integrable processes in $\mathbb{R}^d$ with $L^2=L^2(0,T;\mathbb{R}^d).$

By combing desirable features from SV models and LV models one obtains the following general class of SLV models
\begin{equation}\label{eq:sec2_1}
	\left\{
	\begin{array}{l}
		d\nu_t = \na(t, \nu_t)\,dt + \nb(t, \nu_t)\,dW_{t}^{\nu},\\
		dS_t = \sa(t, S_t, \nu_t)\,dt + \alpha(\nu_t)\beta(S_t)\,dW_{t}^{s},
	\end{array}\right.
\end{equation}
where $W_{t}^{\nu}$ and $W_{t}^{s}$ are correlated with a constant correlation coefficient $\rho \in (-1,1),$ i.e., $d<W^{\nu},W^s>_t=\rho\,dt.$ Instead of using constant correlation we extend the SLV models \eqref{eq:sec2_1} by incorporating stochastic correlation between $W_{t}^{\nu}$ and $W_{t}^{s},$ as
\begin{equation}\label{eq:sec2_2}
	\left\{
	\begin{array}{l}
		d\rho_t = \ra(t, \rho_t)\,dt + \rb(t, \rho_t)\,dW_{t}^{\rho},\\
		d\nu_t = a_{\nu}(t, \nu_t)\,dt + \nb(t, \nu_t)\,dW_{t}^{\nu},\\
		dS_t = \sa(t, S_t, \nu_t,\rho_t)\,dt + \alpha(\nu_t)\beta(S_t)\,dW_{t}^{s},
	\end{array}\right.
\end{equation}
where $d<W^{\nu},W^s>_t=\mathcal{R}_{\nu s}(\rho_t)\,dt.$ This symbolic notation means that
\begin{equation}\label{eq:sec2_3}
	E[W_t^{\nu}W_t^s]=E\left[\int_0^t\mathcal{R}_{s\nu}(\rho_s)\,ds\right],
\end{equation}
we refer to \cite{Teng2016a, Teng2016c} for its construction. Furthermore, we let $d<W^{\rho},W^{\nu}>_t=\mathcal{R}_{\rho \nu}(\rho_t)\,dt,$ and $W_t^{\rho}$ and $W_t^{s}$ are thus correlated in a way which can be interpreted with $\mathcal{R}_{\rho \nu}(\rho_t)$ and $\mathcal{R}_{\nu s}(\rho_t).$ For this reason, we assume that $d<W^{\rho},W^s>_t=\mathcal{R}_{\rho s}(\mathcal{R}_{\rho \nu}(\rho_t),\mathcal{R}_{\nu s}(\rho_t))\,dt.$ And we assume that $\mathcal{R}_{\rho \nu}(\cdot), \mathcal{R}_{\nu s}(\cdot):(-1,1)\to(-1,1)$ and $\mathcal{R}_{\rho s}(\cdot,\cdot):(-1,1)\times(-1,1)\to (-1,1),$ they are defined such that the family of correlation matrices
\begin{equation}\label{eq:sec2_4}
	\mathcal{R}_t=\left(
	\begin{array}{ccc}
		1 & \mathcal{R}_{\rho\nu}(\rho_t) &\mathcal{R}_{\rho s}(\mathcal{R}_{\rho\nu}(\rho_t),\mathcal{R}_{\nu s}(\rho_t))\\
		\mathcal{R}_{\rho\nu}(\rho_t) & 1 &\mathcal{R}_{\nu s}(\rho_t)\\
		\mathcal{R}_{\rho s}(\mathcal{R}_{\rho\nu}(\rho_t),\mathcal{R}_{\nu s}(\rho_t)) & \mathcal{R}_{\nu s}(\rho_t)&1
	\end{array}
	\right),\quad t \geq 0,
\end{equation}
is positive semi-definite.
For the processes $\rho_t$ and $\nu_t,$ we assume that $\ra(t, \rho_t), \rb(t, \rho_t)$ and $\na(t, \nu_t), \nb(t, \nu_t)$ are $L^2$-measurable in $(t, \nu)\in [0,T]\times\mathbb{R}_{\geq 0}$ and $(t, \rho)\in [0,T]\times (-1,1),$ respectively. And there exist positive constants $C_\rho^1, C_\rho^2$ and $C_{\nu}^1, C_\nu^1$ such that
\begin{equation}\label{eq:sec2_5}
	\begin{array}{l}
		|a_{\rho}(t, \rho_1)|^2+|b_{\rho}(t, \rho_1)|^2\leq C_\rho^1(1+\rho_1^2),\\
		|a_{\nu}(t, \nu_1)|^2+|b_{\nu}(t, \nu_1)|^2\leq C_\nu^1(1+\nu_1^2),\\
		|a_{\rho}(t, \rho_1)-a_{\rho}(t, \rho_2)|+|b_{\nu}(t, \rho_1)-b_{\rho}(t, \rho_2)|\leq C_\rho^2|\rho_1-\rho_2|,\\
		|a_{\nu}(t, \nu_1)-a_{\nu}(t, \nu_2)|+|b_{\nu}(t, \nu_1)-b_{\nu}(t, \nu_2)|\leq C_\nu^2|\nu_1-\nu_2|,
	\end{array}
\end{equation}
for any two values $\rho_1, \rho_2 \in (-1,1),$ and any two values $\nu_1, \nu_2 \in \mathbb{R}_{\geq 0}.$ These conditions guarantee that there exist unique solutions for the processes $\nu_t$ and $\rho_t.$ For the underlying asset process $S_t,$ we assume that $\sa(\cdot,\cdot,\cdot,\cdot):\mathbb{R}_{\geq 0}\times\mathbb{R}_{\geq 0}\times\mathbb{R}_{\geq 0}\times (-1,1)\to\mathbb{R}, \alpha(\cdot):\mathbb{R}_{\geq 0}\to \mathbb{R}_{\geq 0}$ and $\beta(\cdot):\mathbb{R}_{\geq 0}\to \mathbb{R}_{\geq 0}$ are continuously differentiable and defined such that \eqref{eq:sec2_2} has a unique solution.
The model class \eqref{eq:sec2_2} is called stochastic correlation and local volatility (SCLV) throughout this paper.

A simple and common choice of correlation structure is to set that $\mathcal{R}_{\rho\nu}(\rho_t) \equiv 0, \mathcal{R}_{\nu s}(\rho_t) \equiv \rho_t$ and $\mathcal{R}_{\rho s}(0,\rho_t) \equiv \rho_{\rho s},$ where $\rho_{\rho s}$ is a constant in $(-1,1)$. The family of correlation matrices reduces thus to
\begin{equation}\label{eq:sec2_6}
	\left(
	\begin{array}{ccc}
		1 & 0 &\rho_{\rho s}\\
		0 & 1 &\rho_t\\
		\rho_{\rho s}& \rho_t&1
	\end{array}
	\right),
\end{equation}
which is positive semi-definite for $\rho_{\rho s}^2 + \rho_t^2 \leq 1$ and $t\geq 0.$ In this work we consider a more general correlation structure by setting
\begin{equation}\label{eq:sec2_7}
	\mathcal{R}_{\rho\nu}(\rho_t) \equiv \rho_{\rho\nu}, \mathcal{R}_{\nu s}(\rho_t) \equiv \rho_t~\mbox{and}~\mathcal{R}_{\rho s}(\rho_{\rho\nu},\rho_t) \equiv \rho_{\rho\nu}\rho_t,
\end{equation}
where $\rho_{\rho \nu}$ is a constant in $(-1,1).$ This correlation structure corresponds the family of correlation matrices as
\begin{equation}\label{eq:sec2_8}
	\left(
	\begin{array}{ccc}
		1 & \rho_{\rho\nu} &\rho_{\rho\nu}\rho_t\\
		\rho_{\rho\nu} & 1 &\rho_t\\
		\rho_{\rho\nu}\rho_t& \rho_t&1
	\end{array}
	\right).
\end{equation}
The positive semi-definiteness of \eqref{eq:sec2_8} can be guaranteed if it holds that
\begin{equation}\label{eq:sec2_9}
	1-\rho_{\rho\nu}^2-(\rho_{\rho\nu}\rho_t)^2-\rho_t^2+2\rho_{\rho\nu}\rho_{\rho\nu}\rho_t\rho_t=(1-\rho_t^2)(1-\rho_{\rho\nu}^2)\geq 0,
\end{equation}
which is true as long as $\rho_{\rho\nu}, \rho_t \in (-1,1).$ A similar correlation structure is discussed in \cite{Higgins2014}. With the aid of Cholesky-decomposition one can reformulate \eqref{eq:sec2_2} as
\begin{equation}\label{eq:sec2_10}
	\left\{
	\begin{aligned}
		d\rho_t &= \ra(t, \rho_t)\,dt + \rb(t, \rho_t)\,d\tbm{t}{\rho},\\
		d\nu_t &= a_{\nu}(t, \nu_t)\,dt + \nb(t, \nu_t)\left(\rho_{\rho \nu}d\tbm{t}{\rho}+\sqrt{1-\rho_{\rho \nu}^2}d\tbm{t}{\nu}\right),\\
		dS_t &= \sa(t, S_t, \nu_t,\rho_t)\,dt + \alpha(\nu_t)\beta(S_t)\left(\rho_{\rho \nu}\rho_t\,d\tbm{t}{\rho}+\sqrt{1-\rho_{\rho \nu}^2}\rho_t\,d\tbm{t}{\nu}+\sqrt{1-\rho_t^2}\,d\tbm{t}{s}\right),
	\end{aligned}
	\right.
\end{equation}
where the Brownian motions $\tbm{t}{\rho},\tbm{t}{\nu}$ and $\tbm{t}{s}$ are independent. The model class \eqref{eq:sec2_10} consolidates and integrates many various SLV models in the literature extended with stochastic correlation. We refer to \cite{Cui2018} for a comprehensive overview of SLV models in Table 1. Furthermore, a variety of SCPs are summarized in Table \ref{tab:sec2_1}.
\begin{table}[h!]
	\caption{Some SCPs in the literature}\label{tab:sec2_1}
	\begin{tabular}{|l|l|l|l|}
		\hline
		$\ra(t,\rho_t)$&$\rb(t,\rho_t)$&$\rho_t \in$ & Parameter restriction\\
		\hline
		\multirow{2}{*}{$\kappa_\rho(\mu_\rho-\rho_t)$}&\multirow{2}{*}{$\sigma_\rho\sqrt{(1-\rho_t^2)}$}&\multirow{2}{*}{$(-1,\,1)$}&$\kappa_\rho\geq\frac{\sigma_\rho^2}{1 \pm \mu_\rho}$\\
		&&&\cite{Emmerich2007}\\
		\hline
		\multirow{3}{*}{$\kappa_\rho(\mu_\rho-\rho_t)$}&\multirow{3}{*}{$\sigma_\rho\sqrt{(\overline{\rho}-\rho_t)(\rho_t-\underline{\rho})}$\tablefootnote{This SCP is called bounded Jacobi process, and reduces to the SCP in the first row by setting $\underline{\rho}=-1$ and $\overline{\rho}=1.$}}&\multirow{3}{*}{$(\underline{\rho},\, \overline{\rho})$}&$1\geq \overline{\rho}\geq \underline{\rho}\geq -1$ and $\overline{\rho}>\kappa_\rho\mu_\rho>\underline{\rho}$\\
		&&&$\kappa_\rho\mu_\rho-\kappa_\rho\underline{\rho}>\frac{\sigma_\rho^2}{2(\overline{\rho}-\underline{\rho})}$ and $\kappa_\rho\overline{\rho}-\kappa_\rho\mu_\rho>\frac{\sigma_\rho^2}{2(\overline{\rho}-\underline{\rho})}$\\
		&&&\cite{Ma2009}\\
		\hline
		$(1-\rho_t^2)(\kappa_\rho(\mu_\rho-$&\multirow{2}{*}{$\sigma_\rho(1-\rho_t^2)$}&\multirow{2}{*}{$(-1,\,1)$}& No.\\
		$\artanh(\rho_t))-\sigma_\rho^2\rho_t)$&&&\cite{Teng2016c}\\
		\hline
		\multirow{2}{*}{$\kappa_\rho(1-\rho_t^2)(\mu_\rho-\rho_t)$}&\multirow{2}{*}{$\sigma_\rho(1-\rho_t^2)$}&\multirow{2}{*}{$(-1,\,1)$}&No.\\
		&&&\cite{Teng2016a}\\
		\hline
		\multirow{2}{*}{$\kappa_\rho(\mu_\rho-\rho_t)$}&\multirow{2}{*}{$\sigma_\rho\sqrt{\rho_t(1-\rho_t)}$}&\multirow{2}{*}{$(0,1)$}&$\kappa_\rho\mu_\rho>\sigma_\rho^2/2$ and $\kappa_\rho(1-\mu_\rho)>\sigma_\rho^2/2$\\
		&&&\cite{Driessen2016}\\
		\hline
	\end{tabular}
\end{table}
Note that it holds $\kappa_\rho, \sigma_\rho>0$ and $\rho_0, \mu_\rho \in(-1,1)$ for all the SCPs in Table \ref{tab:sec2_1}.

\section{Reflected Forward Backward Stochastic Differential Equations for pricing}\label{sec:RBSDEs}
The decoupled reflected forward-backward stochastic differential equations (RBSDEs) for pricing American options exhibit the form of \cite{Karoui1997b}
\begin{equation}\label{eq:sec3_1}
	\left\{
	\begin{array}{l}
		\,\,\, dX_t = a(t, X_t)\,dt + b(t, X_t)\,dW_t,\quad X_0=x_0,\\
		-dY_t = f(t, X_t, Y_t, Z_t)\,dt - Z_t\,dW_t+dL_t,\\
		\quad Y_T=g(X_T), Y_t\geq g(X_t) ~\mbox{for}~t<T~\mbox{and}~\int_{0}^{T}(Y_t-g(X_t))\,dL_t=0,
	\end{array}\right.
\end{equation}
where $a:[0, T]\times \mathbb{R}^d \to \mathbb{R}^d,$ $b:[0, T]\times \mathbb{R}^d \to \mathbb{R}^{d\times d}$, $f: [0, T] \times \mathbb{R}^d \times \mathbb{R} \times \mathbb{R}^{d} \to \mathbb{R}$ are continuous. $f$ is called the driver function or generator of the RBSDE, and $g$ is a known payoff function. $L_t$ is an adapted non-decreasing process, this external force keeps $Y_t$ staying above the lower barrier process $g(X_t).$ Furthermore, the process $f(t,X_t,Y_t,Z_t)$ is $\mathcal{F}_t$-adapted, and $g(X_T)$ is $\mathcal{F}_T$-measurable. A triplet of processes $(Y_t,Z_t,L_t)\in\mathbb{R}\times\mathbb{R}^d\times\mathbb{R}$ is the solution of the RBSDE \eqref{eq:sec3_1} if it is $\mathcal{F}_t$-adapted, square integrable and satisfies \eqref{eq:sec3_1} as
\begin{equation}\label{eq:sec3_17}
	Y_t = g(X_T) + \int_{t}^T f(s, X_s, Y_s, Z_s)\,ds - \int_t^T Z_s\,dW_s+L_T-L_t,\quad t \in [0, T].
\end{equation}
These solutions exist uniquely under Lipschitz conditions, see \cite{Karoui1997a}. Suppose that the terminal value $Y_T$ is of the form $g(X^{t,x}_T),$ where $X^{t,x}_T$ denotes the solution of $dX_t$ in \eqref{eq:sec3_1} starting from $x$ at time $t,$ the same notation applies to $(Y^{t,x}_T, Z^{t,x}_T)$ and $L^{t,x}_T.$ Note that $L^{t,x}_t=0,$ and for $t \leq\widehat{t}\leq T,$ the solution $(Y^{t,x}_{\widehat{t}}, Z^{t,x}_{\widehat{t}})$ of the RBSDE \eqref{eq:sec3_1} can be represented as
\begin{equation*}
	Y^{t,x}_{\widehat{t}} = u(\widehat{t}, X^{t,x}_{\widehat{t}}), \quad Z^{t,x}_{\widehat{t}}=b^{\top}({\widehat{t}}, X^{t,x}_{\widehat{t}}) \nabla u({\widehat{t}}, X^{t,x}_{\widehat{t}}),
\end{equation*}
$\forall t \in [0, T),$ where $\mathds{1}_{(\cdot)}$ is an indicator function, and $u$ is solution of the obstacle problem for the semi-linear parabolic PDE of the form \cite{Karoui1997a}
\begin{equation}\label{eq:sec3_19}
	\max\left(\frac{\partial u}{\partial t} + \underbrace{\sum_i^na_i \partial_i u + \frac{1}{2}\sum_{i,j}^n (bb^T)_{i,j} \partial^2_{i,j} u}_{:=\mathcal{L}u} + f(t, x, u, (\nabla u)b), \, g(x)-u\right)=0
\end{equation}
with the terminal condition $u(T,x)=g(x).$
In turn, suppose that $(Y, Z, L)$ is the solution of RBSDE, $u(t, x)=Y^{t,x}_t$ is a viscosity solution to the obstacle problem.

For the model in \eqref{eq:sec2_10}, $X_t$ in \eqref{eq:sec3_1} is given by $\left(\rho_t, \nu_t, S_t\right)^\top.$ Given the payoff function of a financial derivative, $g(\cdot), Y_t$ represents the value of the derivative, typically referred to as $V(t,\rho_t,\nu_t,S_t).$ And
\begin{equation}\label{eq:sec3_5}
	Z_t=\left(\left(\rb+\nb\rho_{\rho \nu}+\alpha\beta\rho_{\rho \nu}\rho_t\right)\frac{\partial V}{\partial \rho},\left(\nb\sqrt{1-\rho_{\rho \nu}^2}+\alpha\beta\sqrt{1-\rho_{\rho \nu}^2}\rho_t\right)\frac{\partial V}{\partial \nu},\alpha\beta\sqrt{1-\rho_t}\frac{\partial V}{\partial S}\right)
\end{equation}
which presents the hedging strategies with respect to correlation, volatility and spot price, respectively.

For pricing American options in the SCLV model \eqref{eq:sec2_10}, we assume that the correlation risk can be eliminated by including a risk premium, i.e., the price of correlation risk, defined as $\lambda_\rho.$ Accordingly, the drift of the correlation process $\rho_t$ in \eqref{eq:sec2_10} is adjusted to $\ra(t,\rho_t)-\lambda_\rho \rho_t,$ ensuring that $\rho_t$ evolves under a risk-neutral measure. A similar adjustment is made for the volatility process $\nu_t,$ using a volatility risk premium $\lambda_\nu.$ It is important to note that there is no unique way to specify the parameters $\lambda_\rho$ and $\lambda_\nu,$ since the market does not provide tradable assets for directly hedging correlation risk or volatility risk. To formalize this modelling choice, we define the set of all equivalent local martingale measures by
\begin{equation*}
	\{\mathbb{Q}^{\lambda} \sim P \,|\, dW_t^{\mathbb{Q}^{\lambda}}=dW_t+\lambda\,dt~\mbox{and}~\lambda=(\lambda_S,\lambda_{\nu},\lambda_{\rho})\,\mbox{admissible}\},
\end{equation*}
where $\lambda$ denotes the vector of market prices of risks for the underlying asset, volatility, and correlation, respectively. We assume that $\lambda$ is progressively measurable and satisfies Novikov's condition. Our framework selects one specific measure $\mathbb{Q}^{\lambda}$ from this family such that the resulting RBSDE admit a unique solution. Under the chosen measure $\mathbb{Q}^{\lambda},$ the spot price process $S_t$ evolves with a risk-neutral drift $rS_t,$ where r denotes the risk-free interest rate.

We start with the derivation of pricing RBSDEs for American options in the SCLV model \eqref{eq:sec2_10}. To be more specific, we use the following variant of \eqref{eq:sec2_10}
\begin{equation}\label{eq:sec3_6}
	\left\{
	\begin{aligned}
		d\rho_t &= \ra(t, \rho_t)\,dt + \rb(t, \rho_t)\,d\bW{t}{\rho},\\
		d\nu_t &= a_{\nu}(t, \nu_t)\,dt + \nb(t, \nu_t)\left(\rho_{\rho \nu}d\bW{t}{\rho}+\sqrt{1-\rho_{\rho \nu}^2}d\bW{t}{\nu}\right),\\
		dS_t &= \sa(t, S_t)\,dt + \alpha(\nu_t)\beta(S_t)\left(\rho_{\rho \nu}\rho_t\,d\bW{t}{\rho}+\sqrt{1-\rho_{\rho \nu}^2}\rho_t\,d\bW{t}{\nu}+\sqrt{1-\rho_t^2}\,d\bW{t}{s}\right),
	\end{aligned}
	\right.
\end{equation}
where $\bW{t}{\rho}, \bW{t}{\nu}$ and $\bW{t}{s}$ are independent.
In the following proposition, we present the pricing RBSDEs for American options using the SCLV model as defined in \eqref{eq:sec3_6}.
\begin{pro}
	Suppose that $V(t,\rho_t,\nu_t,S_t), t\in[0,T]$ is the price of an American option with the payoff function $g(\cdot),$ and valuated in \eqref{eq:sec3_6}. Then $V(t,\rho_t,\nu_t,S_t)$ can be obtained by solving the RBSDE
	\begin{equation}\label{eq:sec3_20}
		\left\{
		\begin{array}{l}
			dX_t :=\\
			\begin{pmatrix}
				d\rho_t \\
				d\nu_t \\
				dS_t
			\end{pmatrix}=\begin{pmatrix}
				\ra(t, \rho_t) \\
				\na(t, \nu_t) \\
				\sa(t, S_t)
			\end{pmatrix}\,dt+ \begin{pmatrix}
				\rb(t,\rho_t) &0 &0\\
				\rho_{\rho\nu}\nb(t,\nu_t)&\sqrt{1-\rho_{\rho\nu}^2}\nb(t,\nu_t)&0 \\
				\rho_{\rho\nu}\alpha(\nu_t)\beta(S_t)\rho_t&\sqrt{1-\rho_{\rho\nu}^2}\alpha(\nu_t)\beta(S_t)\rho_t&\alpha(\nu_t)\beta(S_t)\sqrt{1-\rho_t^2}
			\end{pmatrix}\,dW_t,\\
			-dY_t = f(t, X_t, Y_t, Z_t)\,dt - Z_t\,dW_t + dL_t,\\
			\quad Y_T=g(S_T), Y_t \geq g(S_t)~\mbox{for}~t<T~\mbox{and}~\int_{0}^{T}(Y_t-g(X_t))\,dL_t=0,
		\end{array}\right.
	\end{equation}
	where $W_t=\left(\bW{t}{\rho},\bW{t}{\nu},\bW{t}{s}\right)^\top,$ $V(t,\rho_t,\nu_t,S_t)=Y_t,$ $Z_t=\left(Z_{1,t},Z_{2,t},Z_{3,t}\right)$ with
	\begin{equation}\label{eq:sec3_21}
		\begin{aligned}
			Z_{1,t}&=\rb(t,\rho_t)\frac{\partial V}{\partial \rho}+\nb(t,\nu_t)\rho_{\rho \nu}\frac{\partial V}{\partial \nu}+\alpha(\nu_t)\beta(S_t)\rho_{\rho \nu}\rho_t\frac{\partial V}{\partial S},\\
			Z_{2,t}&=\nb(t,\nu_t)\sqrt{1-\rho_{\rho \nu}^2}\frac{\partial V}{\partial \nu}+\alpha(\nu_t)\beta(S_t)\sqrt{1-\rho_{\rho \nu}^2}\rho_t\frac{\partial V}{\partial S},\\
			Z_{3,t}&=\alpha(\nu_t)\beta(S_t)\sqrt{1-\rho_t^2}\frac{\partial V}{\partial S},
		\end{aligned}
	\end{equation}
	\begin{equation}
		L_t = \left(\mathcal{L}g + f + \lambda_\rho\rho_t\deriv{V}{\rho}+\lambda_\nu\nu_t\deriv{V}{\nu}+(\sa(t,S_t)-rS_t)\deriv{V}{S}\right)^{-}\mathds{1}_{\{Y_t=g(S_t)\}},
	\end{equation}
	and
	\begin{equation}\label{eq:sec3_22}
		\begin{aligned}
			f(t,X_t,Y_t,Z_t)=&-rY_t-\frac{\lambda_\rho\rho_t}{\rb(t,\rho_t)}Z_{1,t}-\left(\frac{\lambda_\nu\nu_t}{\nb(t,\nu_t)\sqrt{1-\rho_{\rho\nu}^2}}-\frac{\lambda_\rho\rho_t\rho_{\rho\nu}}{\rb(t,\rho_t)\sqrt{1-\rho_{\rho\nu}^2}}\right)Z_{2,t}\\
			&-\left(\frac{\sa(t,S_t)-rS_t}{\alpha(\nu_t)\beta(S_t)\sqrt{1-\rho_t^2}}-\frac{\lambda_\nu\nu_t\rho_t}{\nb(t,\nu_t)\sqrt{1-\rho_t^2}}\right)Z_{3,t}.
		\end{aligned}
	\end{equation}
\end{pro}
\begin{proof}
	Given the aforementioned assumptions, the SCLV model \eqref{eq:sec3_6} can be reformulated under the risk-neutral measure as
	\begin{equation}\label{eq:sec3_10}
		\left\{
		\begin{aligned}
			d\rho_t &= (\ra(t, \rho_t)-\lambda_\rho\rho_t)\,dt + \rb(t, \rho_t)\,d\bW{t}{*\rho},\\
			d\nu_t &= (a_{\nu}(t, \nu_t)-\lambda_\nu\nu_t)\,dt + \nb(t, \nu_t)\left(\rho_{\rho \nu}d\bW{t}{*\rho}+\sqrt{1-\rho_{\rho \nu}^2}d\bW{t}{*\nu}\right),\\
			dS_t &= rS_t\,dt + \alpha(\nu_t)\beta(S_t)\left(\rho_{\rho \nu}\rho_t\,d\bW{t}{*\rho}+\sqrt{1-\rho_{\rho \nu}^2}\rho_t\,d\bW{t}{*\nu}+\sqrt{1-\rho_t^2}\,d\bW{t}{*s}\right),
		\end{aligned}
		\right.
	\end{equation}
	where the Brownian motions $\bW{t}{*\rho},\bW{t}{*\nu}$ and $\bW{t}{*s}$ under the risk-neural measure are independent. Let $V(t,\rho_t,\nu_t,S_t)$ be the value of an American option within the model \eqref{eq:sec3_10} with a payoff function $g(\cdot),$ which can be represented as the solution of the obstacle problem
	\begin{equation*}
		\max\left(\frac{\partial V}{\partial t} + \mathcal{L}V -rV, \, g-V\right)=0
	\end{equation*}
	with the terminal condition $V(T,\rho_T,\nu_T,S_T)=g(S_T),$ where
	\begin{equation}\label{eq:operator}
		\begin{aligned}
			\mathcal{L}V:=&(\ra(t, \rho_t)-\lambda_\rho\rho_t)\deriv{V}{\rho}+(a_{\nu}(t, \nu_t)-\lambda_\nu\nu_t)\deriv{V}{\nu}+rS_t\deriv{V}{S}+\frac{1}{2}\rb^2(t, \rho_t)\deriv[2]{V}{\rho}\\
			&+\frac{1}{2}\nb^2(t, \nu_t)\deriv[2]{V}{\nu}+\frac{1}{2}\alpha^2(\nu_t)\beta^2(S_t)\deriv[2]{V}{S}+\rho_{\rho \nu}\rb(t, \rho_t)\nb(t, \nu_t)\frac{\partial^2 V}{\partial\rho\partial\nu}\\
			&+\rho_{\rho \nu}\rb(t, \rho_t)\rho_t\alpha(\nu_t)\beta(S_t)\frac{\partial^2 V}{\partial\rho\partial S}
			+\nb(t, \nu_t)\rho_t\alpha(\nu_t)\beta(S_t)\frac{\partial^2 V}{\partial\nu\partial S}.
		\end{aligned}
	\end{equation}
	It can be shown that $V$ satisfies the following partial differential equation (PDE)
	\begin{equation}\label{eq:sec3_11}
		\deriv{V}{t} + \mathcal{L}V  - rV + (\mathcal{L}g -rg)^{-}\mathds{1}_{\{g(S)=V(t,\rho, \nu, S)\}}=0,
	\end{equation}
	where $\mathcal{L}$ is defined in \eqref{eq:operator}.
	Let $Y_t$ be an adapted process and $Y_t=V(t,\rho_t,\nu_t,S_t),$ where $\rho_t,\nu_t$ and $S_t$ are given in \eqref{eq:sec3_6}. Applying It\^o lemma one obtains
	\begin{equation}\label{eq:sec3_12}
		\begin{aligned}
			dY_t&=\left(\deriv{V}{t}+\ra(t, \rho_t)\deriv{V}{\rho}+a_{\nu}(t, \nu_t)\deriv{V}{\nu}+\sa(t,S_t)\deriv{V}{S}+\frac{1}{2}\rb^2(t, \rho_t)\deriv[2]{V}{\rho}+\frac{1}{2}\nb^2(t,\nu_t)\deriv[2]{V}{\nu}\right.\\
			&+\frac{1}{2}\alpha^2(\nu_t)\beta^2(S_t)\deriv[2]{V}{S}+\rho_{\rho \nu}\rb(t, \rho_t)\nb(t, \nu_t)\frac{\partial^2 V}{\partial\rho\partial\nu}+\rho_{\rho \nu}\rb(t, \rho_t)\rho_t\alpha(\nu_t)\beta(S_t)\frac{\partial^2 V}{\partial\rho\partial S}\\
			&\left.+\nb(t, \nu_t)\rho_t\alpha(\nu_t)\beta(S_t)\frac{\partial^2 V}{\partial\nu\partial S}\right)\,dt+\left(\rb(t, \rho_t)\deriv{V}{\rho}+\rho_{\rho \nu}\nb(t, \nu_t)\deriv{V}{\nu}+\rho_{\rho \nu}\alpha(\nu_t)\beta(S_t)\rho_t\deriv{V}{S}\right)\,d\bW{t}{\rho}\\
			&+\left(\sqrt{1-\rho_{\rho \nu}^2}\nb(t,\nu_t)\deriv{V}{\nu}+\sqrt{1-\rho_{\rho \nu}^2}\alpha(\nu_t)\beta(S_t)\rho_t\deriv{V}{S}\right)\,d\bW{t}{\nu}+\sqrt{1-\rho_{t}^2}\alpha(\nu_t)\beta(S_t)\deriv{V}{S}\,d\bW{t}{S}.
		\end{aligned}
	\end{equation}
	We compare the coefficient of $dt$ in the latter equation to the PDE \eqref{eq:sec3_11} and rewrite it as
	\begin{equation}\label{eq:sec3_13}
		\begin{aligned}
			-dY_t=&\left(-rY_t-\lambda_\rho\rho_t\deriv{V}{\rho}-\lambda_\nu\nu_t\deriv{V}{\nu}-(\sa(t,S_t)-rS_t)\deriv{V}{S}\right)\,dt+(\mathcal{L}g -rg)^{-}\mathds{1}_{\{g(S_t)=V(t,\rho_t, \nu_t, S_t)\}}dt\\
			&+\left(\rb(t, \rho_t)\deriv{V}{\rho}+\rho_{\rho \nu}\nb(t, \nu_t)\deriv{V}{\nu}+\rho_{\rho \nu}\alpha(\nu_t)\beta(S_t)\rho_t\deriv{V}{S}\right)\,d\bW{t}{\rho}\\
			&+\left(\sqrt{1-\rho_{\rho \nu}^2}\nb(t,\nu_t)\deriv{V}{\nu}+\sqrt{1-\rho_{\rho \nu}^2}\alpha(\nu_t)\beta(S_t)\rho_t\deriv{V}{S}\right)\,d\bW{t}{\nu}+\sqrt{1-\rho_{t}^2}\alpha(\nu_t)\beta(S_t)\deriv{V}{S}\,d\bW{t}{S}.
		\end{aligned}
	\end{equation}
	By defining the non-decreasing c\`adl\`ag process $L_t:=(\mathcal{L}g -rg)^{-}\mathds{1}_{\{g(S_t)=V(t,\rho_t, \nu_t, S_t)\}},$ setting $Z_{1,t},Z_{2,t}$ and $Z_{3,t}$ equal to the values given in \eqref{eq:sec3_21} and rewriting
	\begin{equation}\label{eq:sec3_14}
		f(t,X_t,Y_t,Z_t):= -rY_t-\lambda_\rho\rho_t\deriv{V}{\rho}-\lambda_\nu\nu_t\deriv{V}{\nu}-(\sa(t,S_t)-rS_t)\deriv{V}{S}
	\end{equation}
	with respect to $Z_{1,t},Z_{2,t}$ and $Z_{3,t},$ we complete the proof.
\end{proof}
\begin{rmk}\label{rem:fr}
	The RBSDE in \eqref{eq:sec3_13} is obtained through a derivation from the model presented in \eqref{eq:sec3_6}, considering the historical, real-world $\mathds{P}$-measure. In this context, it becomes evident that there is no necessity to seek the risk-neutral measure for the spot price dynamic for pricing when employing the RBSDE approach. Assume that the model in \eqref{eq:sec3_6} is already under the risk-neutral measure, the corresponding RBSDE thus reads
	\begin{equation}\label{eq:sec3_15}
		\begin{aligned}
			&-dY_t=-rY_t\,dt+dL_t+\left(\rb(t, \rho_t)\deriv{V}{\rho}s+\rho_{\rho \nu}\nb(t, \nu_t)\deriv{V}{\nu}+\rho_{\rho \nu}\alpha(\nu_t)\beta(S_t)\rho_t\deriv{V}{S}\right)\,d\bW{t}{\rho}\\
			&+\left(\sqrt{1-\rho_{\rho \nu}^2}\nb(t,\nu_t)\deriv{V}{\nu}+\sqrt{1-\rho_{\rho \nu}^2}\alpha(\nu_t)\beta(S_t)\rho_t\deriv{V}{S}\right)\,d\bW{t}{\nu}+\sqrt{1-\rho_{t}^2}\alpha(\nu_t)\beta(S_t)\deriv{V}{S}\,d\bW{t}{S},
		\end{aligned}
	\end{equation}
	which is a special case of \eqref{eq:sec3_13} where $\lambda_\rho=\lambda_\nu=0$ and $\sa(t,S_t)=rS_t.$
\end{rmk}
\section{Regression-based numerical schemes}\label{sec:scheme}
For the previously developed RBSDE pricing framework, we construct a gradient boosting-based backward approximation algorithm. The main idea is to first discretize the time integrands using, for example, the general theta-discretization method \cite{Zhao2009, Zhao2012}, and then approximate the resulting conditional expectations via regression-based techniques. While related regression-based algorithms for standard BSDEs have been proposed in \cite{Teng2021, Teng2022}, the numerical treatment of RBSDEs considered here is fundamentally different due to the presence of the reflection mechanism and the associated obstacle constraint. Consequently, both the semi-discrete and fully discrete schemes developed in the following section require a different construction from those for standard BSDEs. Moreover, the convergence analysis for the RBSDE schemes is new and substantially more challenging, as it must explicitly account for the additional errors induced by the reflection and obstacle approximation. We therefore provide a dedicated rigorous convergence analysis for the proposed RBSDE discretizations.
\subsection{The Semi-discrete scheme} We introduce the time partition for the time interval $[0, T]$
\begin{equation}\label{eq:sec4_1}
	\Delta_t = \{t_i | t_i \in [0, T], i=0,1,\cdots,N_T, t_i<t_{i+1}, t_0=0, t_{N_T}=T \}.
\end{equation}
Let $\Delta t_i = t_{i+1}-t_i$ be the time step, and denote the maximum time step with $\Delta t.$ We need to firstly discretize the forward SDE in \eqref{eq:sec3_1}
\begin{equation}\label{eq:sec4_2}
	X_t = x_0 + \int_0^t a(s, X_s)\,ds + \int_0^t b(s, X_s)\,dW_s,
\end{equation}
suppose that discretization is such that
\begin{equation*}\label{eq:sec4_3}
	E\left[\max_{t_i}\left|X_{t_i}-X^{\Delta_t}_{t_i}\right|^2\right]=\mathcal{O}({\Delta t})
\end{equation*}
for the discretized process $X^{\Delta_t}_{t_i}.$ In cases where the process $X_t$ follows a known distribution, this distribution can be used to generate high-quality samples on the grid $\Delta_t;$ otherwise, the Euler scheme may be employed.

For the backward process in \eqref{eq:sec3_1}, we represent its solution in \eqref{eq:sec3_17} as
\begin{equation}\label{eq:scheme1}
	Y_i = Y_{i+1} + \int_{t_i}^{t_{i+1}} f(s, X_s, Y_s, Z_s)\,ds - \int_{t_i}^{t_{i+1}} Z_s\,dW_s+L_{i+1}-L_i,
\end{equation}
where $Y_i:=Y_{t_i}$ and $Y_{i+1}:=Y_{t_{i+1}},$ the same notation is used for the other variables. For \eqref{eq:scheme1} without reflections, we multiply both sides by $\Delta W_{i+1}:=W_{t_{i+1}}-W_{t_{i}}$ and take the conditional expectation $E_i[\cdot]:=E[\cdot|\mathcal{F}_{t_i}],$ yielding
\begin{equation}\label{eq:scheme2}
	E_i[Y_{i+1}\Delta W_{i+1}] = \Delta t_i Z_i + R_{i}^{Z},
\end{equation}
where $R^{Z}_{i}$ denotes the discretization error. Similarly, by taking the conditional expectation of \eqref{eq:scheme1} without reflections we define
\begin{equation}\label{eq:scheme3}
	\widetilde{Y}_i=E_i[Y_{i+1}] + \Delta t_i f(t_i,X_i,\widetilde{Y}_i,Z_i) + R_{i}^{\widetilde{Y}},
\end{equation}
with discretization error denoted by $R_{i}^{\widetilde{Y}}.$ Next, by comparing the result with the barrier imposed by the reflection condition, we define
\begin{equation}\label{eq:scheme4}
	Y_i = \max\{\widetilde{Y}_i,g(X_i)\}.
\end{equation}
The equations \eqref{eq:scheme1},\eqref{eq:scheme2} and \eqref{eq:scheme3} lead to a time discrete approximation $(Y^{\Delta_t},Z^{\Delta_t})$ of $(Y,Z)$ using following semi-discrete scheme:
\begin{algorithm}[H]
	\caption*{\bf Semi-discrete RBSDE Scheme}
	\begin{algorithmic}
		\STATE{\begin{align}
				Y_{N_T}^{\Delta_t} &= g(X_{N_T}^{\Delta_t}),\label{eq:sec4_4}\\
				~\mbox{\bf For}~i&=N_T-1,\cdots,0: \nonumber\\
				Z_i^{\Delta_t}&=\frac{1}{\Delta t_i}E_i[Y_{i+1}^{\Delta_t}\Delta W_{i+1}],\label{eq:sec4_5}\\
				\widetilde{Y}_i^{\Delta_t}&=E_i[Y_{i+1}^{\Delta_t}] + \Delta{t_i} f(t_i, X_i^{\Delta_t},\widetilde{Y}_i^{\Delta_t},Z_i^{\Delta_t}),\label{eq:sec4_6}\\
				Y_i^{\Delta_t}&=\max\left(\widetilde{Y}_i^{\Delta_t}, g(X_i^{\Delta_t})\right).\label{eq:sec4_6_1}
		\end{align}}
	\end{algorithmic}
\end{algorithm}
Note that \eqref{eq:sec4_6} is implicit and thus iterative methods such as Picard iteration may be applied. A similar discretization of RBSDE is considered in \cite{Ma2005}, where a rigorous convergence rate analysis is provided. This convergence result is further recovered in \cite{Bouchard2008}, which shows that
\begin{equation}
	\max_{i=0,\cdots,N_T-1}E\left|Y_i-Y_i^{\Delta_t}\right|^2=\mathcal{O}(\Delta t^{\frac{1}{2}})
\end{equation}
under assumption (H1): $g\in C_b^{1}$ with $L$-Lipschitz derivatives, and
\begin{equation}\label{eq:diserror1}
	\max_{i=0,\cdots,N_T-1}E\left|Y_i-Y_i^{\Delta_t}\right|^2=\mathcal{O}(\Delta t)
\end{equation}
under assumption (H2): $b\in C_b^{1}$ with $L$-Lipschitz derivatives, and $g\in C_b^{2}$ with $L$-Lipschitz first and second derivatives. Furthermore, it holds that
\begin{equation}\label{eq:diserror2}
	E\left[\sum_{i=0}^{N_T-1}\int_{t_i}^{t_{i+1}}\left|Z_t-Z_i^{\Delta_t}\right|^2dt\right]=\mathcal{O}(\Delta t^{\frac{1}{2}}).
\end{equation}
We direct the reader to \cite{Bouchard2008} for a detailed analysis of the time-discrete approximation scheme presented above.
\subsection{Fully discrete scheme}\label{sec:fullyscheme}
The fully discrete scheme results from the computation of conditional expectations within the semi-discrete scheme, and is expressed as follows:
\begin{algorithm}[H]
	\caption*{\bf Fully discrete RBSDE Scheme}
	\begin{algorithmic}
		\STATE{\begin{align*}
				\widehat{y}_{N_T, \mathcal{M}} &= g(\widehat{x}_{N_T, \mathcal{M}}),\\
				~\mbox{\bf For}~i&=N_T-1,\cdots,0~,~\mathcal{M}=1,\cdots, M:\\
				\widehat{z}_{i, \mathcal{M}}&=E^{\widehat{x}_{i, \mathcal{M}}}_{i}\left[\frac{1}{\Delta t_i}\widehat{Y}^{\Delta_t}_{i+1}\Delta W_{i+1} \right],\\
				\widehat{\widetilde{y}}_{i, \mathcal{M}}&=E^{\widehat{x}_{i, \mathcal{M}}}_{i}\left[\widehat{Y}^{\Delta_t}_{i+1} \right] + \Delta{t_i} f(\widehat{x}_{i,\mathcal{M}},\widehat{\widetilde{y}}_{i, \mathcal{M}},\widehat{z}_{i, \mathcal{M}}),\\
				\widehat{y}_{i, \mathcal{M}}&=\max\left(\widehat{\widetilde{y}}_{i, \mathcal{M}}, g(\widehat{x}_{i, \mathcal{M}})\right).
		\end{align*}}
	\end{algorithmic}
\end{algorithm}
In the fully discrete scheme above, $\widehat{Y}^{\Delta_t}_{i+1}$ denote a random variable whose distribution is represented by the samples $\widehat{y}_{i+1, \mathcal{M}}, \mathcal{M}=1,\cdots,M.$
\paragraph{Non-parameteric regression} As mentioned in the preceding section, we apply a regression-based method to compute the conditional expectation. Specifically, we employ the eXtreme Gradient Boosting (XGBoost) regression algorithm \cite{Chen2016}, a machine learning technique known for its high predictive performance. More precisely, the XGBoost non-parametric regression approach is used to estimate conditional expectations of the form $E[Y|X=x]$ based on the model
\begin{equation}
	Y = \eta(X) + \epsilon,
\end{equation}
where $\epsilon$ has zero mean, bounded variance, and is independent of $X.$ It then follows that
\begin{equation}\label{eq:conditioneta}
	E[Y|X=x]=\eta(x).
\end{equation}
To approximate this conditional expectation, our goal in regression is to construct an estimator of the function, $\widehat{\eta}(x).$ In non-parametric regression, we do not assume a specific functional form for $\eta.$ Instead, $\widehat{\eta}$ is represented by an XGBoost regressor. Suppose we have a dataset, $(\widehat{x}_\mathcal{M}, \widehat{y}_\mathcal{M}),\,\mathcal{M}=1,\cdots,M,$ sampled from $(X, Y).$ We split the data into training and test sets and fit the model, namely, XGBoost regressor. on the training data. The fitted regressor can then be used to predict $E[Y|X=x]$ both in-sample and out-of-sample.

The use of the XGBoost algorithm for approximating conditional expectations in the context of BSDEs has been described in detail in \cite{Teng2022}. We adopt this method for the conditional expectations in the fully discrete RBSDE scheme, which is more challenging. We assume that the process $(X_{i}^{\Delta_t})_{i=0,\cdots,N_t}$ is Markovian, there thus exist deterministic functions such that
\begin{equation}\label{eq:TheFun}
	\widetilde{Y}_i^{\Delta_t} = y_i^{\Delta_t} (X_i^{\Delta_t})~\mbox{and}~Z_i^{\Delta_t} = z^{\Delta_t}_i (X_i^{\Delta_t}),
\end{equation}
where $\widetilde{Y}_i^{\Delta_t}$ and $Z_i^{\Delta_t}$ are defined in \eqref{eq:sec4_6} and \eqref{eq:sec4_5}, respectively, in the semi-discrete RBSDE scheme. In the fully discrete RBSDE scheme, $E^{\widehat{x}_{i, \mathcal{M}}}_{i}[\mathcal{S}_{i+1}]$ denotes the predicted conditional expectation $E[\mathcal{S}_{i+1}|X=\widehat{x}_{i, \mathcal{M}}],$ obtained using the XGBoost model fitted to the dataset of $\mathcal{S}_{i+1},$ denoted by $\widehat{\mathcal{S}}_{i+1,\mathcal{M}},$ where $\mathcal{M}=1,\cdots, M,$ and $M$ is the sample size. The variable $\mathcal{S}_{i+1}$ can be either $\frac{1}{\Delta t_i}\widehat{Y}^{\Delta_t}_{i+1}\Delta W_{i+1}$ or $\widehat{Y}^{\Delta_t}_{i+1},~\mbox{for}~i=N_T-1,\cdots, 0.$

Taking $\mathcal{S}_{i+1}=\frac{1}{\Delta t_i}\widehat{Y}^{\Delta_t}_{i+1}\Delta W_{i+1}$ and its sample $\widehat{\mathcal{S}}_{i+1,\mathcal{M}}:=\frac{1}{\Delta t_{i}}\widehat{y}_{i+1,\mathcal{M}}\Delta \widehat{w}_{i+1,\mathcal{M}}$ as an example, we start at time $T$ by fitting the regressor $\widehat{R}^{z}_{N_T-1}(x)$ to approximate the conditional expectation $E\left[\mathcal{S}_{N_T}|X^{\Delta_t}_{N_T-1}=x\right]$ using the dataset $\left(\widehat{x}_{N_T-1, \mathcal{M}}, \widehat{\mathcal{S}}_{N_T,\mathcal{M}}\right).$ This implies that the function $z^{\Delta t}_{N_T-1}(x)$ defined in \eqref{eq:TheFun} is estimated and represented by the regressor, i.e.,
\begin{equation}\label{eq:regressor}
	z^{\Delta_t}_{N_T-1}(x) \approx \widehat{R}^{z}_{N_T-1}(x),
\end{equation}
which can be used to predict the conditional expectation $E^{\widehat{x}_{N_T-1, \mathcal{M}}}_{N_T-1}\left[\mathcal{S}_{N_T} \right]$ both in-sample and out-of-sample. The predicted values of $E^{\widehat{x}_{N_T-1, \mathcal{M}}}_{N_T-1}\left[\mathcal{S}_{N_T}\right]$ are denoted by $\widehat{z}_{N_T-1, \mathcal{M}},$ which can be interpreted as the generated sample values of $Z^{\Delta_t}_{N_T-1}.$ Recursively, proceeding backward in time, the dataset $\widehat{z}_{N_T-1, \mathcal{M}}$ (and also $\widehat{y}_{N_T-1, \mathcal{M}}$) is used to construct the dataset $\widehat{z}_{N_T-2, \mathcal{M}}$ corresponding to the random variables $Z_{N_T-2}^{\Delta_t}$ at time $t_{N_T-2}.$
At the initial time $t=0,$ we are given a fix initial value $x_0$ for $X,$ i.e., this dataset is deterministic.

For the $i$-th step, we define the approximation errors as
\begin{equation}\label{eq:err_z1}
	R_{\mbox{xgb}}^{Z_i}:=\frac{1}{M}\sum_{\mathcal{M}=1}^M \left(\widehat{z}_{i, \mathcal{M}}-z^{\Delta_t}_i(\widehat{x}_{i,\mathcal{M}})\right)^2.
\end{equation}
Note that $\widehat{\mathcal{S}}_{i+1, \mathcal{M}}=z^{\Delta_t}_i(\widehat{x}_{i,\mathcal{M}})+\epsilon_{i,\mathcal{M}}^{z},$ where the noise term $\epsilon^{z}_i$ has zero mean and a bounded constant variance $Var^{z}_i$. Hence, the expression in \eqref{eq:err_z1} can be reformulated as
\begin{align}
	R_{\mbox{xgb}}^{Z_i}&=\frac{1}{M}\sum_{\mathcal{M}=1}^M \left(\epsilon_{i,\mathcal{M}}^{z}-\left(\widehat{\mathcal{S}}_{i+1,\mathcal{M}}-\widehat{z}_{i, \mathcal{M}}\right)\right)^2\\
	&\leq 2Var^{z}_i+\frac{2}{M}\sum_{\mathcal{M}=1}^M \left(\widehat{\mathcal{S}}_{i+1,\mathcal{M}}-\widehat{z}_{i, \mathcal{M}}\right)^2,\label{eq:err_z}
\end{align}
where the second term represents the XGBoost regression error, which will be analyzed in the sequel.\\
Similarly, for the case $\mathcal{S}_{i+1}=\widehat{Y}^{\Delta_t}_{i+1}$ and $\widehat{\mathcal{S}}_{i+1, \mathcal{M}}=y^{\Delta_t}_i(\widehat{x}_{i,\mathcal{M}})+\epsilon_{i,\mathcal{M}}^{y},$
where $y^{\Delta_t}(\cdot)$ is defined in \eqref{eq:TheFun} and the noise term $\epsilon_{i}^{y}$ has zero mean and a bounded constant variance $Var^{y}_i,$ we define
\begin{equation}\label{eq:err_y}
	R_{\mbox{xgb}}^{Y_i}:=\frac{1}{M}\sum_{\mathcal{M}=1}^M \left(\widehat{y}_{i, \mathcal{M}}-y^{\Delta_t}_i(\widehat{x}_{i,\mathcal{M}})\right)^2,
\end{equation}
where the predicted values are obtained from the XGBoost regressor. We then deduce the bound
\begin{equation}\label{eq:err_y_1}
	R_{\mbox{xgb}}^{Y_i}\leq 2Var^{y}_i+\frac{2}{M}\sum_{\mathcal{M}=1}^M \left(\widehat{\mathcal{S}}_{i+1,\mathcal{M}}-\widehat{y}_{i, \mathcal{M}}\right)^2,
\end{equation}
under the assumption that errors from the iterative method can be neglected by choosing a sufficiently large number of Picard iterations.

\paragraph{Regularized learning objective}
For notational simplicity, we omit the time-step index, e.g., we write $\widehat{x}_{\mathcal{M}}$ instead of $\widehat{x}_{i,\mathcal{M}}.$
Suppose we are given a dataset with $M$ samples
\begin{equation*}
	\mathcal{D}=\left\{(\widehat{x}_{\mathcal{M}},\widehat{\mathcal{S}}_{\mathcal{M}})|~|\mathcal{D}|=M\right\}.
\end{equation*}
We continue to use the approximation for the $Z$-component as an example. A tree ensemble model consisting of $K$ regression trees can then be constructed to predict the output
\begin{equation*}
	\widehat{z}_{\mathcal{M}}=\eta(\widehat{x}_\mathcal{M})=\sum_{k=1}^{K}\xf_k(\widehat{x}_{\mathcal{M}}),\quad \xf_k \in \mathcal{S},
\end{equation*}
where $\mathcal{S}:=\{\xf(x)=\omega_{q(x)},\,q:\mathbb{R}\to \widehat{T},\, \omega \in \mathbb{R}^{\widehat{T}}\}$ is the space of regression trees, and $q$ denotes the structure of each tree, mapping an input $x$ to the corresponding leaf index, and $\omega_j$ represents the score associated with the $j$-th leaf. Each $\xf_k$ corresponds to an independent tree structure $q$ and leaf weights $\omega.$ The total number of leaves is denoted by $\widehat{T}.$

To train the model, we minimize the mean squared error (MSE)
\begin{equation}\label{eq:mse}
	L(\widehat{z}_{\mathcal{M}}, \widehat{\mathcal{S}}_{\mathcal{M}}) = \frac{1}{M}\sum_{\mathcal{M}=1}^{M}(\widehat{z}_{\mathcal{M}}-\widehat{\mathcal{S}}_{\mathcal{M}})^2
\end{equation}
for regression. For a regularized objective we define the regularization term as
\begin{equation*}
	\Omega(\xf) = \gamma \widehat{T} + \frac{1}{2}\lambda \norm{w}^2=\gamma \widehat{T} + \frac{1}{2}\lambda \sum_{j=1}^{\widehat{T}}w_j^2,
\end{equation*}
where $\gamma$ and $\lambda$ are positive regularization parameters, and $w_j$ denotes the score on the $j$-th leaf. This regularization term controls the complexity of the model and helps prevent overfitting. The fully regularized objective is then given by
\begin{equation}\label{eq:objective}
	\mathcal{L}(\eta)= \sum_{\mathcal{M}=1}^M 	L(\widehat{z}_{\mathcal{M}}, \widehat{\mathcal{S}}_{\mathcal{M}}) + \sum_{k=1}^K\Omega(f_k),
\end{equation}
which needs to be minimized. Here, $L$ serves as the loss function that quantifies the discrepancy between the predictions and targets.
\paragraph{Gradient Tree Boosting}
In XGBoost, gradient descent is employed to minimize \eqref{eq:objective}. Accordingly, we minimize the following objective by sequentially adding functions $\xf_k$ in an iterative algorithm
\begin{equation}\label{eq:objetivet}
	\begin{split}
		\mathcal{L}^{(k)} &= \sum_{\mathcal{M}=1}^{M}L(\widehat{\mathcal{S}}_{\mathcal{M}},\widehat{z}_{\mathcal{M}}^{k})+\sum_{j=1}^k\Omega(\xf_j)\\&=\sum_{\mathcal{M}=1}^{M}L(\widehat{\mathcal{S}}_{\mathcal{M}},\widehat{z}_{\mathcal{M}}^{(k-1)}+\xf_k(\widehat{x}_{\mathcal{M}}))+\Omega(\xf_k),
	\end{split}
\end{equation}
where $\widehat{z}_{\mathcal{M}}^{k}=\sum_{j=1}^{k}\xf_j(\mathbf{x_{\mathcal{M}}}),$ and $k=1,\cdots, K.$ That is, \eqref{eq:objective} is minimized by greedily adding $\xf_k.$ To this end, a second-order approximation of \eqref{eq:objetivet} is computed as
\begin{equation*}
	\mathcal{L}^{(k)} \approx \sum_{\mathcal{M}=1}^{M}\left(L(\widehat{\mathcal{S}}_{\mathcal{M}},\widehat{z}_{\mathcal{M}}^{(k-1)})+g_{\mathcal{M}}\xf_k(\widehat{x}_{\mathcal{M}})+\frac{1}{2}h_{\mathcal{M}}\xf^2_k(\widehat{x}_{\mathcal{M}})\right)+\Omega(\xf_k),
\end{equation*}
where $g_{\mathcal{M}}=\partial_{\widehat{z}^{(k-1)}}L(\widehat{\mathcal{S}}_{\mathcal{M}},\widehat{z}^{(k-1)})$ and $h_{\mathcal{M}}=\partial^2_{\widehat{z}^{(k-1)}}L(\widehat{\mathcal{S}}_{\mathcal{M}},\widehat{z}^{(k-1)})$ denote the first- and second-order gradients, respectively. By removing constant terms, one obtains the simplified objective at the $k$-th step
\begin{equation}\label{eq:aobjetivet}
	\mathcal{\widetilde{L}}^{(k)} = \sum_{\mathcal{M}=1}^{M}\left(g_{\mathcal{M}}\xf_k(\widehat{x}_{\mathcal{M}})+\frac{1}{2}h_{\mathcal{M}}\xf^2_k(\widehat{x}_{\mathcal{M}})\right)+\Omega(\xf_k),
\end{equation}
which must be optimized by determining a suitable $\xf_k.$

Next, we show how to construct the tree $\xf_k$ that minimizes this objective. We first define a tree as
\begin{equation*}
	\xf_k(\widehat{x}) = w_{q(\widehat{x})}, \quad w \in \mathbb{R}^{\widehat{T}},
\end{equation*}
and introduce $I_j=\{\mathcal{M}|q(\widehat{x}_{\mathcal{M}})=j\}$ as the instance set of leaf $j,$ which contains the indices of data points mapped to the $j$-th leaf. Then, \eqref{eq:aobjetivet} can be reformulated as
\begin{align*}
	\mathcal{\widetilde{L}}^{(k)} &= \sum_{\mathcal{M}=1}^{M}\left(g_{\mathcal{M}}\xf_k(\widehat{x}_{\mathcal{M}})+\frac{1}{2}h_{\mathcal{M}}\xf^2_k(\widehat{x}_{\mathcal{M}})\right)+\gamma \widehat{T} + \frac{1}{2}\lambda\sum_{j=1}^{\widehat{T}}w_j^2\\
	&=\sum_{j=1}^{\widehat{T}}\left(\left(\sum_{\mathcal{M}\in I_j}g_{\mathcal{M}}\right)w_j+\frac{1}{2}\left(\sum_{\mathcal{M}\in I_j}h_{\mathcal{M}}+\lambda\right)w_j^2\right)+\gamma \widehat{T}.
\end{align*}
For a fixed $q(\widehat{x}),$ the optimal optimal $w_j$ for each leaf $j$ can be computed explicitly as
\begin{equation*}
	w^*_{j}=-\frac{\sum_{\mathcal{M}\in I_j}g_{\mathcal{M}}}{\sum_{\mathcal{M}\in I_j}h_{\mathcal{M}}+\lambda},
\end{equation*}
which leads to the corresponding optimal value of the objective
\begin{equation}\label{eq:xgberror}
	\mathcal{\widetilde{L}}^{(k)}(q)	=-\frac{1}{2}\sum_{j=1}^{\widehat{T}}\frac{(\sum_{\mathcal{M}\in I_j}g_{\mathcal{M}})^2}{\sum_{\mathcal{M}\in I_j}h_{\mathcal{M}}+\lambda}+\gamma\widehat{T}.
\end{equation}
This value can serve as a scoring function to evaluate the quality of a given tree structure $q.$ Due to the high computational cost, it is infeasible to exhaustively enumerate all possible tree structures $q.$ A greedy algorithm, proposed in \cite{Chen2016}, recursively identifies best splitting points until a specified maximum depth is reached. We denote the instance sets of the left and right child nodes after a split by $I_L$ and $I_R$, respectively, with $I=I_L \cup I_R.$ The loss reduction after the split is given by
\begin{equation}
	\mathcal{L}_{\mbox{gain}}=\frac{1}{2}\left(\frac{(\sum_{\mathcal{M}\in I_L}g_{\mathcal{M}})^2}{\sum_{\mathcal{M}\in I_L}h_{\mathcal{M}}+\lambda}+\frac{(\sum_{\mathcal{M}\in I_R}g_{\mathcal{M}})^2}{\sum_{\mathcal{M}\in I_R}h_{\mathcal{M}}+\lambda}-\frac{(\sum_{\mathcal{M}\in I}g_{\mathcal{M}})^2}{\sum_{\mathcal{M}\in I}h_{\mathcal{M}}+\lambda}\right)-\gamma,
\end{equation}
where $\lambda$ is the regularization parameter. This reduction is used to evaluate candidate splits. Based on the gain, nodes with negative loss reduction are pruned. We have now introduced the mathematical formulation underlying XGBoost. For additional techniques implemented to further prevent overfitting and to improve computational efficiency, we refer the reader to \cite{Chen2016}. With the loss functions introduced above, the errors of approximation conditional expectations in \eqref{eq:err_z} and \eqref{eq:err_y_1} can be further reformulated as
\begin{equation}\label{eq:zyerrors}
	R_{\mbox{xgb}}^{Z_i}\leq 2Var^{z}_i+2\widehat{\mathcal{L}}_{\mbox{min}}(\widehat{q}^z_i)~\mbox{and}~R_{\mbox{xgb}}^{Y_i}\leq 2Var^{y}_i+2\widehat{\mathcal{L}}_{\mbox{min}}(\widehat{q}^{y}_i),
\end{equation}
where $\widehat{\mathcal{L}}_{\mbox{min}}(\cdot)$ is defined in \eqref{eq:xgberror}, and $\widehat{q}^{\cdot}_i$ denotes the tree structure that achieves the minimal error at the $i$-th time step.

The time complexity of XGBoost is provided in \cite{Chen2016} as
\begin{equation*}
	\mathcal{O}\left(K\widetilde{d}|\mbox{samples}|+|\mbox{samples}|\log B\right),
\end{equation*}
where $\widetilde{d}$ denotes the maximum depth of the tree, $K$ is the total number of trees, and $|\mbox{samples}|$ is the number of training samples. The parameter $B$ refers to the maximum number of rows in each block; a column block structure is used in the XGBoost algorithm to enable efficient parallel learning. Therefore, we can straightforwardly deduce the complexity of our scheme $\mathcal{O}\left(K\widetilde{d}MdN_T+MdN_T\log B\right).$
\subsection{Convergence analysis}\label{sec:ConAna}
\paragraph{Assumption (H3):} Assumption (H2) holds. Suppose that $X_0$ is $\mathcal{F}_0$-measurable with $E[|X_0|^2]<\infty,$ and that $a$ and $b$ are $L^2$-measurable in $(t,x) \in [0, T]\times \mathbb{R}^d,$ are linear growth bounded and uniformly Lipschitz continuous, i.e., there exist positive constants $K$ and $L$ such that
\begin{align*}
	|a(t,x)|^2 &\leq K(1+|x|^2),\quad |b(t,x)|^2 \leq K(1+|x|^2),\\
	|a(t,x)-a(t,y)| &\leq L|x-y|, \quad |b(t,x)-b(t,y)| \leq L|x-y|
\end{align*}
with $x, y \in \mathbb{R}^d.$

The forward process $X_t$ is numerically approximated, when it can not be simulated based on its distribution. For notational convenience, we omit the dependency of $X_i=X_i^{\Delta_t}$ on the time grid $\Delta_t,$ and the dependency of the local and global errors on the state of $X_i.$
For the implicit schemes we will apply Picard iterations which converges with the Lipschitz assumptions on the driver, and for any initial guess when $\Delta t_i$ is small enough. In the following analysis, we consider the equidistant time discretization. The deterministic functions $y_i^{\Delta_t} (\cdot)~\mbox{and}~z^{\Delta_t}_i (\cdot)$ in \eqref{eq:TheFun} are approximated by the regressors $\widehat{R}^{y}_{i}(\cdot)$ and $\widehat{R}^{z}_{i}(\cdot),$ respectively (see also \eqref{eq:regressor}). We define
$\widehat{\widetilde{Y}}_i^{\Delta_t}:=\widehat{R}^{y}_{i}(X_i^{\Delta_t}),$ $\widehat{Y}_i^{\Delta_t}:=\max\left(\widehat{\widetilde{Y}}_i^{\Delta_t},g(X_i^{\Delta_t})\right)$ and $\widehat{Z}_i^{\Delta_t}:=\widehat{R}^{z}_{i}(X_i^{\Delta_t}),$ and introduce the corresponding approximation errors as
\begin{equation}
	\widehat{\epsilon}^{Y_i}:= Y_i^{\Delta_t}-\widehat{Y}_i^{\Delta_t},~\widehat{\epsilon}^{Z_i}:= Z_i^{\Delta_t}-\widehat{Z}_i^{\Delta_t},~\widehat{\widetilde{\epsilon}}^{Y_i}:= \widetilde{Y}_i^{\Delta_t}-\widehat{\widetilde{Y}}_i^{\Delta_t},
\end{equation}
where $Z_i^{\Delta_t},$ $\widetilde{Y}_i^{\Delta_t}$ and $Y_i^{\Delta_t}$ are defined in \eqref{eq:sec4_5}, \eqref{eq:sec4_6} and \eqref{eq:sec4_6_1}, respectively, in the semi-discrete scheme. As introduced earlier, $\widehat{\widetilde{Y}}_i^{\Delta_t}$ and $\widehat{Z}_i^{\Delta_t}$ denote random variables whose distributions are represented by the generated samples by XGBoost regressors. Thus, we define the global errors as
\begin{equation}
	\epsilon^{Y_i}:= Y_i-\widehat{Y}_i^{\Delta_t},~\epsilon^{Z_i}:= Z_i-\widehat{Z}_i^{\Delta_t},~\widetilde{\epsilon}^{Y_i}:= \widetilde{Y}_i-\widehat{\widetilde{Y}}_i^{\Delta_t},
\end{equation}
where $\widetilde{Y}_i$ and $Y_i$ are defined in \eqref{eq:scheme3} and \eqref{eq:scheme4}, respectively.
\begin{theorem}
	Under Assumption (H3), and noting that $E[\widehat{\epsilon}^{Y_{N_T}}]=0$ due to the known terminal condition, we obtain
	\begin{equation}
		E[|\epsilon^{Y_i}|^2]+ \widetilde{C}\Delta t E[|\epsilon^{Z_i}|^2] \leq \mathcal{O} (\Delta t) + C\sum_{j=i+1}^{N_T} \left( \frac{N_T}{T} Var_j^\mathcal{Y}+\frac{T}{N_T} Var_j^\mathcal{Z}\right),
	\end{equation}
	for all $0 \leq i \leq N_T-1$ and $ 0<\Delta t<1,$ where the constant $\widetilde{C} \in \left(0,\frac{1}{2d}\right),$ and the constant $C$ depends on $T, L,~\mbox{and}~d.$ The quantities $Var_{j}^\mathcal{Y}$ and $Var_{j}^\mathcal{Z}$ represent bounded constants.
\end{theorem}
\begin{proof}
	For the $Z$-component $(0\leq i \leq N_T-1),$ we consider
	\begin{equation}\label{eq:ez}
		\widehat{\epsilon}^{Z_i} = \frac{1}{\Delta t}E_i[Y_{i+1}^{\Delta_t}\Delta W_{i+1}] - \frac{1}{\Delta t}E^{\widehat{x}_i}_i[\widehat{{Y}}_{i+1}^{\Delta_t}\Delta W_{i+1}]=\frac{1}{\Delta t}E_i[\widehat{\epsilon}^{{Y}_{i+1}}\Delta W_{i+1}]+R_{\mbox{xgb}}^{Z_{i}},
	\end{equation}
	where $E^{\widehat{x}_i}_i[\widehat{\mathcal{S}}]$ denotes the approximation of $E^i[\mathcal{S}]$ obtained using an XGBoost regressor trained on the dataset $(\widehat{x}_{i,\mathcal{M}}, \widehat{\mathcal{S}}_\mathcal{M}), \mathcal{M}=1,\cdots, M.$ Consequently, $E^{\widehat{x}_i}_i[\widehat{{Y}}_{i+1}^{\Delta_t}\Delta W_{i+1}]=E_i[Y_{i+1}^{\Delta_t}\Delta W_{i+1}]+R_{\mbox{xgb}}^{Z_{i}},$ where $R_{\mbox{xgb}}^{Z_{i}}$ is the approximation error introduced by the XGBoost model and defined in \eqref{eq:err_z1}. By applying Cauchy-Schwarz inequality, we estimate $E_i[\widehat{\epsilon}^{{Y}_{i+1}}\Delta W_{i+1}]$ in the Frobenius norm as
	\begin{align*}
		\left|E_i[\widehat{\epsilon}^{{Y}_{i+1}}\Delta W_{i+1}]\right|^2&=\left|E_i[(\widehat{\epsilon}^{{Y}_{i+1}}-E_i[\widehat{\epsilon}^{{Y}_{i+1}}])\Delta W_{i+1}]\right|^2\\ &\leq d \Delta t \left(E_i[|\widehat{\epsilon}^{{Y}_{i+1}}|^2]-|E_i[\widehat{\epsilon}^{{Y}_{i+1}}]|^2\right),
	\end{align*}
	which can then be substituted into \eqref{eq:ez} to yield
	\begin{equation}\label{eq:ez2}
		|\widehat{\epsilon}^{Z_i}|^2\leq\frac{2d}{\Delta t}\left(E_i[|\widehat{\epsilon}^{{Y}_{i+1}}|^2]-|E_i[\widehat{\epsilon}^{{Y}_{i+1}}]|^2\right)+2|R_{\mbox{xgb}}^{Z_{i}}|^2.
	\end{equation}
	For the $\widetilde{Y}$-component, we consider
	\begin{align*}
		|\widehat{\widetilde{\epsilon}}^{Y_i}|&\leq \left|\widetilde{Y}_i^{\Delta_t}-\widehat{\widetilde{Y}}_i^{\Delta_t}\right|=\left|E_i[Y_{i+1}^{\Delta_t}]+f(t_i,X_i^{\Delta_t},\widetilde{Y}_i^{\Delta_t},Z_i^{\Delta_t})-E^{\widehat{x}_i}_i[\widehat{Y}_{i+1}^{\Delta_t}]-f(t_i,X_i^{\Delta_t},\widehat{\widetilde{Y}}_i^{\Delta_t},\widehat{Z}_i^{\Delta_t})\right|,\\
		&\leq \left|E_i[\widehat{\epsilon}^{Y_{i+1}}] + \Delta t L (\widehat{\widetilde{\epsilon}}^{Y_{i}}+\widehat{\epsilon}^{Z_{i}})+R_{\mbox{xgb}}^{Y_i}\right|,
	\end{align*}
	where $R_{\mbox{xgb}}^{Y_i}$ given in \eqref{eq:err_y_1}.
	Applying Young's inequality, we estimate
	\begin{align*}
		|\widehat{\widetilde{\epsilon}}^{Y_i}|^2
		&\leq (1+\gamma \Delta t)\left|E_i[\widehat{\epsilon}^{Y_{i+1}}]\right|^2 + (1+\frac{1}{\gamma \Delta t})\left|\Delta t L (\widehat{\widetilde{\epsilon}}^{Y_{i}}+\widehat{\epsilon}^{Z_{i}})+R_{\mbox{xgb}}^{Y_i}\right|^2.\\
		& \leq (1+\gamma \Delta t)\left|E_i[\widehat{\epsilon}^{Y_{i+1}}]\right|^2 + 3L^2\Delta t^2 |\widehat{\widetilde{\epsilon}}^{Y_i}|^2+ 3L^2\Delta t^2 |\widehat{\epsilon}^{Z_i}|^2 + 3|R_{\mbox{xgb}}^{Y_i}|^2\\
		&~~~~~~~~~~~~~+ \frac{1}{\gamma \Delta t}\left(3L^2\Delta t^2 |\widehat{\widetilde{\epsilon}}^{Y_i}|^2+ 3L^2\Delta t^2 |\widehat{\epsilon}^{Z_i}|^2 + 3|R_{\mbox{xgb}}^{Y_i}|^2\right),
	\end{align*}
	for some $\gamma >0$ to be chosen later. Consequently, by collecting coefficients, we obtain
	\begin{align*}
		\left(1-3L^2\Delta t^2-\frac{3L^2\Delta t}{\gamma}\right)|\widehat{\widetilde{\epsilon}}^{Y_i}|^2
		& \leq (1+\gamma \Delta t)\left|E_i[\widehat{\epsilon}^{Y_{i+1}}]\right|^2 + 3L^2\Delta t^2 |\widehat{\epsilon}^{Z_i}|^2 + 3|R_{\mbox{xgb}}^{Y_i}|^2\\
		&~~~~~~~~+ \frac{1}{\gamma \Delta t}\left(3L^2\Delta t^2 |\widehat{\epsilon}^{Z_i}|^2 + 3|R_{\mbox{xgb}}^{Y_i}|^2\right).
	\end{align*}
	Since
	\begin{equation*}
		|\widehat{\epsilon}^{Y_i}|=\left|\max\left(\widetilde{Y}_i^{\Delta_t},g(X_i^{\Delta_t})\right)-\max\left(\widehat{\widetilde{Y}}_i^{\Delta_t},g(X_i^{\Delta_t})\right)\right|\leq\left|\widetilde{Y}_i^{\Delta_t}-\widehat{\widetilde{Y}}_i^{\Delta_t}\right|=|\widehat{\widehat{\epsilon}}^{Y_i}|,
	\end{equation*}
	$|\widehat{\widetilde{\epsilon}}^{Y_i}|$ can be replaced by $|\widehat{\epsilon}^{Y_i}|,$ leading to
	\begin{equation}\label{eq:ey}
		\begin{split}
			\left(1-3L^2\Delta t^2-\frac{3L^2\Delta t}{\gamma}\right)|\widehat{\epsilon}^{Y_i}|^2
			& \leq (1+\gamma \Delta t)\left|E_i[\widehat{\epsilon}^{Y_{i+1}}]\right|^2 + \left(3L^2\Delta t^2+\frac{3L^2\Delta t}{\gamma}\right) |\widehat{\epsilon}^{Z_i}|^2 \\
			&~~~~~~~~+ \left(\frac{3}{\gamma \Delta t} + 3\right)|R_{\mbox{xgb}}^{Y_i}|^2.
		\end{split}
	\end{equation}
	Using \eqref{eq:ez2} and \eqref{eq:ey}, we compute
	\begin{align*}
		&\left(1-3L^2\Delta t^2-\frac{3L^2\Delta t}{\gamma}\right)|\widehat{\epsilon}^{Y_i}|^2 + \left(\frac{\Delta t}{2d}-3L^2\Delta t^2-\frac{3L^2\Delta t}{\gamma}\right)|\widehat{\epsilon}^{Z_i}|^2\leq (1+\gamma \Delta t)\left|E_i[\widehat{\epsilon}^{Y_{i+1}}]\right|^2\\
		&~~~~~~~~~~~~+ \left(\frac{3}{\gamma \Delta t} + 3\right)|R_{\mbox{xgb}}^{Y_i}|^2+\left(E_i[|\widehat{\epsilon}^{{Y}_{i+1}}|^2]-|E_i[\widehat{\epsilon}^{{Y}_{i+1}}]|^2\right)+\frac{\Delta t}{2d}|R_{\mbox{xgb}}^{Z_{i}}|^2.
	\end{align*}
	We choose $\gamma\geq 6dL^2$ so that $\frac{1}{2d}-\frac{2L^2}{\gamma}\geq 0.$ Hence, there exist generic constants $\widehat{C}$ with $0<\widehat{C}<\frac{1}{2d},$ depending only on $L, \gamma,$ and $d,$ such that the inequality can be written as
	\begin{align*}
		&\left(1-\widehat{C}\Delta t\right)|\widehat{\epsilon}^{Y_i}|^2 + \Delta t \widetilde{C}\left(1-\widehat{C}\Delta t\right)|\widehat{\epsilon}^{Z_i}|^2\leq (1+\widehat{C} \Delta t)E_i[|\widehat{\epsilon}^{{Y}_{i+1}}|^2]\\
		&~~~~~~~~~~~+\left(\frac{3}{\gamma \Delta t} + 3\right)|R_{\mbox{xgb}}^{Y_i}|^2+\frac{\Delta t}{2d}|R_{\mbox{xgb}}^{Z_{i}}|^2.
	\end{align*}
	By induction and taking expectations $E[\cdot]$ on both sides, we obtain
	\begin{align*}
		&E[|\widehat{\epsilon}^{Y_i}|^2] + \Delta t \widetilde{C}E[|\widehat{\epsilon}^{Z_i}|^2]\leq \left(\frac{1+\widehat{C} \Delta t}{1-\widehat{C} \Delta t}\right)^{N_T-i}E[|\widehat{\epsilon}^{{Y}_{N_T}}|^2]\\
		&~~~~~~~~~~~~~~~+\sum_{j=i+1}^{N_T}\left(\frac{1}{1-\widehat{C} \Delta t}\right)^{j-i}\left(\left(\frac{3+3\gamma \Delta t}{\gamma \Delta t} \right)|R_{\mbox{xgb}}^{Y_j}|^2+\frac{\Delta t}{2d}|R_{\mbox{xgb}}^{Z_{j}}|^2\right)\\
		&~~~~~~~~~~~\leq \exp(\widehat{C}T) \sum_{j=i+1}^{N_T}\left(\left(\frac{3+3\gamma \Delta t}{\gamma \Delta t} \right)|R_{\mbox{xgb}}^{Y_j}|^2+\frac{\Delta t}{2d}|R_{\mbox{xgb}}^{Z_{j}}|^2\right).
	\end{align*}
	Based on \eqref{eq:zyerrors}, we define
	\begin{equation*}
		Var_j^{\mathcal{Z}}:= (2Var^{z}_i+2\widehat{\mathcal{L}}_{\mbox{min}}(\widehat{q}^z_i))^2~\mbox{and}~Var_j^{\mathcal{Y}}:= (2Var^{y}_i+2\widehat{\mathcal{L}}_{\mbox{min}}(\widehat{q}^{y}_i))^2,
	\end{equation*}
	and use these quantities to rewrite the inequality as
	\begin{align*}
		E[|\widehat{\epsilon}^{Y_i}|^2] + \Delta t \widetilde{C}E[|\widehat{\epsilon}^{Z_i}|^2]\leq C\sum_{j=i+1}^{N_T} \left( \frac{N_T}{T} Var_j^\mathcal{Y}+\frac{T}{N_T} Var_j^\mathcal{Z}\right),
	\end{align*}
	where $C$ depends on $L, \gamma, d,$ and $T.$ Finally, incorporating the known discretization errors defined in \eqref{eq:diserror1} and \eqref{eq:diserror2}, we complete the proof.
\end{proof}
\section{Applications of RBSDE approach to SCLV models}\label{sec:application}
As examples, we consider the classical SABR and Heston stochastic local volatility (HSLV) models extended by incorporating a stochastic correlation structure, which belong to our SCLV model class defined in \eqref{eq:sec2_10}.
\subsection{Extension of the classical SABR model}\label{sec:extensionSABR}
The classical SABR model reads \cite{Hagan2002}
\begin{equation}\label{eq:sec5_1}
	\left\{
	\begin{array}{l}
		d\nu_t = \sigma_\nu\nu_t\,dW_{t}^{\nu},\\
		dS_t = \nu_t S_t^{\beta}\,dW_{t}^{s},
	\end{array}\right.
\end{equation}
where $W_{t}^{\nu}$ and $W_{t}^{s}$ are allowed to be correlated with a constant correlation coefficient, $\nu_0,\sigma_\nu>0$ and $0\leq\beta<1.$ We extend \eqref{eq:sec5_1} by including stochastic correlation driven by \cite{Teng2016a}
\begin{equation}\label{eq:sec5_2}
	d\rho_t = \kappa_\rho(\mu_\rho-\rho_t)(1-\rho_t^2)\,dt+\sigma_\rho(1-\rho_t^2)\,dW_t^{\rho},
\end{equation}
see also Table \ref{tab:sec2_1}. The extended SABR model can be derived as
\begin{equation}\label{eq:sec5_3}
	\left\{
	\begin{aligned}
		d\rho_t &= \kappa_\rho(\mu_\rho-\rho_t)(1-\rho_t^2)\,dt + \sigma_\rho(1-\rho_t^2)\,d\tbm{t}{\rho},\\
		d\nu_t &= \sigma_\nu\nu_t\left(\rho_{\rho \nu}d\tbm{t}{\rho}+\sqrt{1-\rho_{\rho \nu}^2}d\tbm{t}{\nu}\right),\\
		dS_t &= \nu_tS_t^{\beta}\left(\rho_{\rho \nu}\rho_t\,d\tbm{t}{\rho}+\sqrt{1-\rho_{\rho \nu}^2}\rho_t\,d\tbm{t}{\nu}+\sqrt{1-\rho_t^2}\,d\tbm{t}{s}\right),
	\end{aligned}
	\right.
\end{equation}
which belongs to the class of SCLV models presented in \eqref{eq:sec2_10}, and is referred to as the Stochastic Alpha Beta Stochastic Rho (SABSR) model in this work.

\begin{coro}\label{co:sec5_2}
	Suppose that an option with payoff function $g(S_t)$ is valuated in \eqref{eq:sec5_3}, which can be exercised at any time before the maturity $T,$ the corresponding pricing RBSDE can be derived as
	\begin{equation}\label{eq:sec5_6}
		\left\{
		\begin{array}{l}
			dX_t :=\\
			\begin{pmatrix}
				d\rho_t \\
				d\nu_t \\
				dS_t
			\end{pmatrix}=\begin{pmatrix}
				\kappa_\rho(\mu_\rho-\rho_t)(1-\rho_t^2) \\
				0 \\
				0
			\end{pmatrix}\,dt+ \begin{pmatrix}
				\sigma_\rho(1-\rho_t^2) &0 &0\\
				\rho_{\rho\nu}\sigma_\nu\nu_t&\sqrt{1-\rho_{\rho\nu}^2}\sigma_\nu\nu_t&0 \\
				\rho_{\rho\nu}\nu_tS^\beta_t\rho_t&\sqrt{1-\rho_{\rho\nu}^2}\nu_tS^\beta_t\rho_t&\nu_tS_t^\beta\sqrt{1-\rho_t^2}
			\end{pmatrix}\,dW_t,\\
			-dY_t = f(t, X_t, Y_t, Z_t)\,dt - Z_t\,dW_t + dL_t,\\
			\quad Y_T=g(S_T), Y_t \geq g(S_t)~\mbox{for}~t<T~\mbox{and}~\int_{0}^{T}(Y_t-g(X_t))\,dL_t=0,
		\end{array}\right.
	\end{equation}
	with
	\begin{equation}\label{eq:sec5_7}
		\begin{aligned}
			f(t,X_t,Y_t,Z_t)=&-rY_t-\frac{\lambda_\rho\rho_t}{\sigma_\rho(1-\rho_t^2)}Z_{1,t}-\left(\frac{\lambda_\nu}{\sigma_\nu\sqrt{1-\rho_{\rho\nu}^2}}-\frac{\lambda_\rho\rho_t\rho_{\rho\nu}}{\sigma_\rho(1-\rho_t^2)\sqrt{1-\rho_{\rho\nu}^2}}\right)Z_{2,t}\\
			&+\left(\frac{r}{\nu_tS_t^{\beta-1}\sqrt{1-\rho_t^2}}+\frac{\lambda_\nu\rho_t}{\sigma_\nu\sqrt{1-\rho_t^2}}\right)Z_{3,t}.
		\end{aligned}
	\end{equation}
	where $W_t=\left(\bW{t}{\rho},\bW{t}{\nu},\bW{t}{s}\right)^\top.$
	\begin{proof}
		The proof can be executed straightforwardly done by setting $\ra(t,\rho_{t})=\kappa_\rho(\mu_\rho-\rho_t)(1-\rho_t^2), \rb(t,\rho_t)=\sigma_\rho(1-\rho_t^2), \na(t,\nu_t)=0,\nb(t,\nu_t)=\sigma_\nu\nu_t,\sa(t,S_t)=0,\alpha(\nu_t)=\nu_t$ and $\beta(S_t)=S_t^\beta$ into the RBSDE given in \eqref{eq:sec3_20}.
	\end{proof}
\end{coro}
\subsection{Extension of the Heston SLV model}\label{sec:extensionHestinSLV}
The HSLV model proposed in \cite{VanderStoep2014} is given by
\begin{equation}\label{eq:sec5_8}
	\left\{
	\begin{array}{l}
		d\nu_t = \kappa_\nu(\mu_\nu-\nu_t)\,dt+\sigma_\nu\sqrt{\nu_t}\,dW_{t}^{\nu},\\
		dS_t = \mu_sS_t\,dt+\sqrt{\nu_t} S_t^{\beta}\,dW_{t}^{s},
	\end{array}\right.
\end{equation}
where $W_{t}^{\nu}$ and $W_{t}^{s}$ are allowed to be correlated with a constant correlation coefficient, $\nu_0, \kappa_\nu,\mu_\nu,\sigma_\nu>0,$ $2\kappa_\nu\mu_\nu>\sigma_\nu^2$ and $0\leq \beta<1.$ Incorporating stochastic correlation allows us to extend \eqref{eq:sec5_8} into a SCLV model, presented as
\begin{equation}\label{eq:sec5_9}
	\left\{
	\begin{aligned}
		d\rho_t &= \kappa_\rho(\mu_\rho-\rho_t)(1-\rho_t^2)\,dt + \sigma_\rho(1-\rho_t^2)\,d\tbm{t}{\rho},\\
		d\nu_t &= \kappa_\nu(\mu_\nu-\nu_t)\,dt+\sigma_\nu\sqrt{\nu_t}\left(\rho_{\rho \nu}d\tbm{t}{\rho}+\sqrt{1-\rho_{\rho \nu}^2}d\tbm{t}{\nu}\right),\\
		dS_t &= \mu_sS_t\,dt+\sqrt{\nu_t}S_t^{\beta}\left(\rho_{\rho \nu}\rho_t\,d\tbm{t}{\rho}+\sqrt{1-\rho_{\rho \nu}^2}\rho_t\,d\tbm{t}{\nu}+\sqrt{1-\rho_t^2}\,d\tbm{t}{s}\right),
	\end{aligned}
	\right.
\end{equation}
which is referred to as the Heston stochastic correlation local volatility (HSCLV) model.

\begin{coro}
	Suppose that an option with payoff function $g(S_t)$ is valuated in \eqref{eq:sec5_9}, which can be exercised at any time before the maturity $T,$ the corresponding pricing RBSDE can be derived as
	\begin{equation}\label{eq:sec5_12}
		\left\{
		\begin{array}{l}
			dX_t :=\\
			\begin{pmatrix}
				d\rho_t \\
				d\nu_t \\
				dS_t
			\end{pmatrix}=\begin{pmatrix}
				\kappa_\rho(\mu_\rho-\rho_t)(1-\rho_t^2) \\
				\kappa_\nu(\mu_\nu-\nu_t) \\
				\mu_sS_t
			\end{pmatrix}\,dt+ \begin{pmatrix}
				\sigma_\rho(1-\rho_t^2) &0 &0\\
				\rho_{\rho\nu}\sigma_\nu\nu_t&\sqrt{1-\rho_{\rho\nu}^2}\sigma_\nu\nu_t&0 \\
				\rho_{\rho\nu}\nu_tS^\beta_t\rho_t&\sqrt{1-\rho_{\rho\nu}^2}\nu_tS^\beta_t\rho_t&\nu_tS_t^\beta\sqrt{1-\rho_t^2}
			\end{pmatrix}\,dW_t,\\
			-dY_t = f(t, X_t, Y_t, Z_t)\,dt - Z_t\,dW_t + dL_t,\\
			\quad Y_T=g(S_T), Y_t \geq g(S_t)~\mbox{for}~t<T~\mbox{and}~\int_{0}^{T}(Y_t-g(X_t))\,dL_t=0,
		\end{array}\right.
	\end{equation}
	with
	\begin{equation}\label{eq:sec5_13}
		\begin{aligned}
			f(t,X_t,Y_t,Z_t)=&-rY_t-\frac{\lambda_\rho\rho_t}{\sigma_\rho(1-\rho_t^2)}Z_{1,t}-\left(\frac{\lambda_\nu\sqrt{\nu_t}}{\sigma_\nu\sqrt{1-\rho_{\rho\nu}^2}}-\frac{\lambda_\rho\rho_t\rho_{\rho\nu}}{\sigma_\rho(1-\rho_t^2)\sqrt{1-\rho_{\rho\nu}^2}}\right)Z_{2,t}\\
			&-\left(\frac{\mu_s-r}{\sqrt{\nu_t}S_t^{\beta-1}\sqrt{1-\rho_t^2}}-\frac{\lambda_\nu\sqrt{\nu_t}\rho_t}{\sigma_\nu\sqrt{1-\rho_t^2}}\right)Z_{3,t}.
		\end{aligned}
	\end{equation}
	where $W_t=\left(\bW{t}{\rho},\bW{t}{\nu},\bW{t}{s}\right)^\top.$
	\begin{proof}
		The proof can be executed straightforwardly done by setting $\ra(t,\rho_{t})=\kappa_\rho(\mu_\rho-\rho_t)(1-\rho_t^2), \rb(t,\rho_t)=\sigma_\rho(1-\rho_t^2), \na(t,\nu_t)=	\kappa_\nu(\mu_\nu-\nu_t),\nb(t,\nu_t)=\sigma_\nu\sqrt{\nu_t},\sa(t,S_t)=\mu_sS_t,\alpha(\nu_t)=\sqrt{\nu_t}$ and $\beta(S_t)=S_t^\beta$ into the RBSDE given in \eqref{eq:sec3_20}.
	\end{proof}
\end{coro}

\section{Numerical experiments}\label{sec:num}
In this section, we present numerical results pertaining to the pricing of American options contracts within the SCLV models. The intention is to illustrate the applicability of the proposed RBSDE approach and its efficiency, and to see the role of stochastic correlation on pricing and hedging. In all numerical experiments, we employ the following XGBoost parameter settings unless otherwise specified: the maximum tree depth is set to $\widetilde{d}=2;$ the learning rate to $\eta=0.9$ (with a low tree depth and a large learning rate chosen to expedite computation when possible), datasets of size $M$ are split into training and test sets with a ration of $3:1,$ the number of trees is set to be $150,$ and the ``Early Stopping" technique is applied with \texttt{early\_stopping\_rounds} set to $1.$ For all remaining parameters, default values are employed. We perform $10$ Picard iterations to solve the implicit equation for $Y$-component.

As indicated in \cite{Teng2022}, comparable accuracy levels can be attained across a range of standard learning rates (i.e., between $0$ and $1$) by adjusting the number of trees for general problems. Nevertheless, challenging problems necessitate a smaller learning rate. In \cite{Teng2022}, the number of trees is determined through manual comparison of training and test curves, and this value is applied consistently across all time steps. In contrast to this methodology, we employ ``Early Stopping" to automatically determine the specific number of trees at each time step.

In the case of a standard BSDE (i.e., $dX_t=dW_t$), there is an absence of forward discretization error. Moreover, our SCLV models are formulated using decoupled forward BSDEs, where the underlying processes $X_t$ are forward SDEs. These forward processes need to be numerically simulated, e.g., with the Euler-Maruyama scheme, when their analytical solutions are unavailable.

The algorithms are executed independently ten times, and the average value is computed for approximation. Different seeds are used for each simulation. We denote the approximations with numerical scheme by $(Y_0^{\Delta_t},\mathbf{Z}_0^{\Delta_t}),$ where $\mathbf{Z}_0^{\Delta_t}=(Z_{1,0}^{\Delta_t},Z_{2,0}^{\Delta_t},\cdots,Z_{d,0}^{\Delta_t}),$ and the approximations of delta-hedgings $\frac{\partial V}{\partial S}$ by $\mathit{\Delta}_{d,0}^{\Delta_t}.$ The approximated prices $\widehat{Y}_0^{\Delta_t}:=\frac{1}{10}\sum_{k=1}^{10}Y_{0,k}^{\Delta_t}$ and delta-hedgings $\widehat{\mathit{\Delta}}_{d,0}^{\Delta_t}=\widehat{Z}^{\Delta_t}_{d,0}/\alpha(\nu_0)\beta(S_0)\sqrt{1-\rho_0^2}$ with $\widehat{Z}_{d,0}^{\Delta_t}:=\frac{1}{10}\sum_{k=1}^{10}Z_{d,0,k}^{\Delta_t}$ will be reported for all the examples, as well as their empirical standard deviations $\sqrt{\frac{1}{9}\sum_{k=1}^{10}\left|\widehat{Y}_0^{\Delta_t}-Y_{0,k}^{\Delta_t}\right|^2}$ and $\sqrt{\frac{1}{9}\sum_{k=1}^{10}\left|\widehat{\mathit{\Delta}}_{d,0}^{\Delta_t}-\mathit{\Delta}_{d,0,k}^{\Delta_t}\right|^2},$ where $\mathit{\Delta}_{d,0,k}^{\Delta_t}=Z^{\Delta_t}_{d,0,k}/\alpha(\nu_0)\beta(S_0)\sqrt{1-\rho_0^2}.$
The dynamic of underlying asset price, denoted by $dS_t,$ are specified as the $d$-th (last) forward component in our pricing models. Consequently, delta-hedging consistently carries the subscript $d.$ Although our approximations of $\widehat{Z}_{1,0},\widehat{Z}_{2,0}\cdots,\widehat{Z}_{d-1,0}$ are also accurate, obtaining sufficiently precise approximations for other hedgings, such as $\frac{\partial V}{\partial \nu}$ and $\frac{\partial V}{\partial \rho},$ poses challenges due to ill-conditioned systems of equations, as illustrated in \eqref{eq:sec3_21}. As a result, we limit our discussion only to delta-hedgings in this paper, and deferring the analysis and treatment of these ill-conditioned systems to future research. All experiments are conducted in Python with an Intel(R) Core(TM) i$5-8500$ CPU @ $3.00$GHz and $15$ GB RAM.
\subsection{Numerical Example: American options with the SABR and SABSR models} In this section, we assess our RBSDE scheme by pricing American options under the classical SABR model defined in \eqref{eq:sec5_1}, as well as under the SABSR model - an extension of the SABR framework that incorporates stochastic correlation as specified in \eqref{eq:sec5_3}.
\subsubsection{American options with the SABR model}\label{sec:aocc}
We begin with the classical SABR model, whose corresponding pricing RBSDE can be readily obtained from the SABSR pricing RBSDE in \eqref{eq:sec5_6} by replacing the stochastic correlation with a constant correlation. For the SABR model, several effective numerical methods are available, including the exact simulation scheme of \cite{Cai2017} and the CTMC method of \cite{Cui2018}. In particular, the CTMC method yields highly accurate approximations with relatively low computational cost. To assess the accuracy of our RBSDE scheme, we employ the parameter values reported in Table 4 of Sec.4.2.3 in \cite{Cui2018}. Following \cite{Gobet2005}, we also apply a logarithmic transformation to $S_t$ to mitigate the insufficient smoothness of the payoff function. The convergence results are presented in Table \ref{tab:sec6_7}, with $N_T$ set to $100, 200$ and $400.$
\begin{table}[h!]
	\centering
	\caption{American put options. Option parameters: $S_0=1.1, T=1, \beta=0.7, \nu_0=0.2, \sigma_\nu=0.08, r=0.$ Constant correlation: $-0.4.$ RBSDE parameters: $M=100000.$ CTMC row is reported in \cite{Cui2018}.}\label{tab:sec6_7}
	\resizebox{\columnwidth}{!}{
		\begin{tabular}{c|cccccccccccccc}
			\hline
			\hline
			&&\multicolumn{11}{c}{$K$}\\
			\cline{3-13}
			$N_T$&&$1.00$&$1.02$&$1.04$&$1.06$&$1.08$&$1.10$&$1.12$&$1.14$&$1.16$&$1.18$&$1.20$\\
			\cline{2-13}
			&CTMC&$\mathbf{0.0424}$&$\mathbf{0.0494}$&$\mathbf{0.0572}$&$\mathbf{0.0658}$&$\mathbf{0.0751}$&$\mathbf{0.0851}$&$\mathbf{0.0959}$&$\mathbf{0.1074}$&$\mathbf{0.1196}$&$\mathbf{0.1326}$&$\mathbf{0.1461}$\\
			\cline{1-13}
			\multirow{5}{*}{$100$}&$\widehat{Y}_0$&$0.0403$&$0.0470$&$0.0547$&$0.0630$&$0.0721$&$0.0821$&$0.0929$&$0.1044$&$0.1167$&$0.1297$&$0.1435$\\
			&(Std.Dev.)&$2.91\mbox{E-}{04}$&$3.37\mbox{E-}{04}$&$3.83\mbox{E-}{04}$&$3.89\mbox{E-}{04}$&$4.09\mbox{E-}{04}$&$4.28\mbox{E-}{04}$&$4.19\mbox{E-}{04}$&$4.30\mbox{E-}{04}$&$4.55\mbox{E-}{04}$&$3.98\mbox{E-}{04}$&$4.75\mbox{E-}{04}$\\
			\cline{2-13}
			&$\widehat{\mathit{\Delta}}_{3,0}^{\Delta_t}$&$-0.2564$&$-0.2892$&$-0.3267$&$-0.3619$&$-0.3971$&$-0.4387$&$-0.4765$&$-0.5154$&$-0.5540$&$-0.5921$&$-0.6241$\\
			&(Std. Dev.)&$3.99\mbox{E-}{03}$&$6.84\mbox{E-}{03}$&$7.57\mbox{E-}{03}$&$7.27\mbox{E-}{03}$&$7.94\mbox{E-}{03}$&$9.50\mbox{E-}{03}$&$8.31\mbox{E-}{03}$&$1.15\mbox{E-}{02}$&$1.00\mbox{E-}{02}$&$1.32\mbox{E-}{02}$&$1.29\mbox{E-}{02}$\\
			\cline{2-13}
			& Average running time &\multicolumn{11}{c}{$23$ seconds}\\
			\cline{1-13}
			\multirow{5}{*}{$200$}&$\widehat{Y}_0$&$0.0413$&$0.0482$&$0.0559$&$0.0643$&$0.0735$&$0.0835$&$0.0943$&$0.1059$&$0.1181$&$0.1313$&$0.1450$\\
			&(Std.Dev.)&$1.84\mbox{E-}{04}$&$2.24\mbox{E-}{04}$&$2.27\mbox{E-}{04}$&$2.49\mbox{E-}{04}$&$3.13\mbox{E-}{04}$&$2.72\mbox{E-}{04}$&$3.78\mbox{E-}{04}$&$3.70\mbox{E-}{04}$&$3.63\mbox{E-}{04}$&$3.29\mbox{E-}{04}$&$3.88\mbox{E-}{04}$\\
			\cline{2-13}
			&$\widehat{\mathit{\Delta}}_{3,0}^{\Delta_t}$&$-0.2376$&$-0.2683$&$-0.3038$&$-0.3390$&$-0.3742$&$-0.4097$&$-0.4484$&$-0.4847$&$-0.5204$&$-0.5581$&$-0.5923$\\
			&(Std. Dev.)&$9.25\mbox{E-}{03}$&$1.21\mbox{E-}{02}$&$1.20\mbox{E-}{02}$&$1.59\mbox{E-}{02}$&$1.72\mbox{E-}{02}$&$1.96\mbox{E-}{02}$&$2.32\mbox{E-}{02}$&$2.62\mbox{E-}{02}$&$2.89\mbox{E-}{02}$&$2.90\mbox{E-}{02}$&$3.36\mbox{E-}{02}$\\
			\cline{2-13}
			& Average running time &\multicolumn{11}{c}{$45$ seconds}\\
			\cline{1-13}
			\multirow{5}{*}{$400$}&$\widehat{Y}_0$&$\mathbf{0.0425}$&$\mathbf{0.0496}$&$\mathbf{0.0573}$&$\mathbf{0.0657}$&$\mathbf{0.0750}$&$\mathbf{0.0850}$&$\mathbf{0.0958}$&$\mathbf{0.1073}$&$\mathbf{0.1195}$&$\mathbf{0.1326}$&$\mathbf{0.1464}$\\
			&(Std.Dev.)&$2.87\mbox{E-}{04}$&$2.67\mbox{E-}{04}$&$3.29\mbox{E-}{04}$&$2.83\mbox{E-}{04}$&$3.92\mbox{E-}{04}$&$3.66\mbox{E-}{04}$&$3.80\mbox{E-}{04}$&$3.65\mbox{E-}{04}$&$3.71\mbox{E-}{04}$&$3.81\mbox{E-}{04}$&$2.76\mbox{E-}{04}$\\
			\cline{2-13}
			&$\widehat{\mathit{\Delta}}_{3,0}^{\Delta_t}$&$-0.2214$&$-0.2507$&$-0.2868$&$-0.3228$&$-0.3605$&$-0.3981$&$-0.4345$&$-0.4752$&$-0.5132$&$-0.5502$&$-0.5846$\\
			&(Std. Dev.)&$7.48\mbox{E-}{03}$&$9.34\mbox{E-}{03}$&$1.00\mbox{E-}{02}$&$1.08\mbox{E-}{02}$&$1.24\mbox{E-}{02}$&$1.35\mbox{E-}{02}$&$1.66\mbox{E-}{02}$&$1.79\mbox{E-}{02}$&$2.27\mbox{E-}{02}$&$2.22\mbox{E-}{02}$&$2.60\mbox{E-}{02}$\\
			\cline{2-13}
			& Average running time &\multicolumn{11}{c}{$126$ seconds}\\
			\cline{1-13}
			\hline
	\end{tabular}}
\end{table}

We observe that the results are overall quite promising, particularly as the approximated prices already attain a level of accuracy comparable to the CTMC approximation in \cite{Cui2018} when $N_T=100.$ As mentioned before, the reported approximations of delta-hedging $\frac{\partial V}{\partial S}$ in Table \ref{tab:sec6_7} are computed with $\widehat{Z}_{d,0}/\nu_0S_0^{\beta}.$ In this example, we set $M=100000$ to expedite computations. The mean execution time for each run in our approach varies with the values of $N_T,$ but remains nearly same for repeated runs with the same $N_T$. Table \ref{tab:sec6_7} also reports the approximate execution times corresponding to each value of $N_T.$ The comparison to the benchmark demonstrates that the proposed RBSDE scheme can attain a quite high accuracy with stable convergence.

\subsubsection{American options under the SABSR model}
In this section we evaluate American options in the SABSR model, and analyze the role of stochastic correlation in pricing and hedging. Examining the influence of stochastic correlation on American options pricing is both significant interesting and novel in the literature. Specially, we set $\rho_0=-0.4, \mu_\rho=-0.4, \sigma_{\rho}=0.5, \rho_{\rho \nu}=0,$ ensuring that the SCP initiates at $-0.4$ and maintains a long-term mean of $-0.4.$ The impact of stochastic correlation is demonstrated by comparing prices and delta-hedgings in Table \ref{tab:sec6_8}. These variances arise solely from differing values of $\kappa_{\rho},$ which dictates the speed at which the SCP reverts to its long-term mean of $-0.4.$ Noting that the CTMC method cannot be directly applied to the SABSR model, which requires substantial reformulation - the existing literature offers very few highly accurate classical numerical approaches. Recently, many deep learning-based algorithms have been developed for pricing American options, and some of these can be adopted for valuation under the SABSR framework. In particular, the deep primal-dual BSDE (DPDB) algorithm proposed in the recent work \cite{Yang2024b} provides both lower and upper bounds for American option prices, accompanied by confidence intervals, making it a suitable benchmark for our study. In Table \ref{tab:sec6_8}, we also report the DPDB approximations together with the corresponding $95\%$ confidence intervals for the lower and upper bounds.
\begin{table}[h!]
	\centering
	\caption{American put options. Option parameters: $S_0=1.1, T=1, \beta=0.7, \nu_0=0.2, \sigma_\nu=0.08, r=0.$ SCP: $\rho_0=-0.4, \mu_\rho=-0.4, \sigma_{\rho}=0.5, \rho_{\rho \nu}=0.$ RBSDE parameters: $N_T=400, M=100000.$ The average running time of our scheme is approximately $210$ seconds for the given parameter setting, and the training time of DPDB scheme is about $637$ seconds.}\label{tab:sec6_8}
	\scalebox{1}[1.3]{
		\resizebox{\columnwidth}{!}{
			\begin{tabular}{c|ccccccccccccc}
				\hline
				\hline
				&&\multicolumn{11}{c}{$K$}\\
				\cline{3-13}
				$\kappa_\rho$&&$1.00$&$1.02$&$1.04$&$1.06$&$1.08$&$1.10$&$1.12$&$1.14$&$1.16$&$1.18$&$1.20$\\
				\cline{1-13}
				\multirow{6}{*}{$0.1$}&$\widehat{Y}_0$&$\mathbf{0.0440}$&$\mathbf{0.0510}$&$\mathbf{0.0587}$&$\mathbf{0.0671}$&$\mathbf{0.0763}$&$\mathbf{0.0862}$&$\mathbf{0.0969}$&$\mathbf{0.1084}$&$\mathbf{0.1207}$&$\mathbf{0.1337}$&$\mathbf{0.1473}$\\
				&(Std.Dev.)&$3.08\mbox{E-}{04}$&$2.78\mbox{E-}{04}$&$2.16\mbox{E-}{04}$&$2.84\mbox{E-}{04}$&$2.28\mbox{E-}{04}$&$4.41\mbox{E-}{04}$&$3.42\mbox{E-}{04}$&$2.58\mbox{E-}{04}$&$3.73\mbox{E-}{04}$&$3.67\mbox{E-}{04}$&$4.13\mbox{E-}{04}$\\
				\cline{2-13}&DPDB&$0.0419$&$0.0488$&$0.0569$&$0.0653$&$0.0745$&$0.0842$&$0.0953$&$0.1064$&$0.1187$&$0.1317$&$0.1447$\\
				&$95\%$ CI bounds&$\mathbf{[0.0418,0.0485]}$&$\mathbf{[0.0488,0.0551]}$&$\mathbf{[0.0569,0.0635]}$&$\mathbf{[0.0652,0.0719]}$&$\mathbf{[0.0744,0.0816]}$&$\mathbf{[0.0841,0.0920]}$&$\mathbf{[0.0952,0.1028]}$&$\mathbf{[0.1063,0.1142]}$&$\mathbf{[0.1187,0.1269]}$&$\mathbf{[0.1315,0.1407]}$&$\mathbf{[0.1445,0.1538]}$\\
				\cline{2-13}
				&$\widehat{\mathit{\Delta}}_{3,0}^{\Delta_t}$&$-0.2681$&$-0.3012$&$-0.3345$&$-0.3692$&$-0.4044$&$-0.4389$&$-0.4735$&$-0.5113$&$-0.5398$&$-0.5761$&$-0.6077$\\
				&(Std.Dev.)&$1.50\mbox{E-}{02}$&$1.71\mbox{E-}{02}$&$2.05\mbox{E-}{02}$&$2.46\mbox{E-}{02}$&$2.64\mbox{E-}{02}$&$3.36\mbox{E-}{02}$&$3.63\mbox{E-}{02}$&$3.72\mbox{E-}{02}$&$4.32\mbox{E-}{02}$&$4.85\mbox{E-}{02}$&$5.35\mbox{E-}{02}$\\                 		   \cline{1-13}
				\multirow{6}{*}{$2.0$}&$\widehat{Y}_0$&$\mathbf{0.0441}$&$\mathbf{0.0509}$&$\mathbf{0.0585}$&$\mathbf{0.0669}$&$\mathbf{0.0762}$&$\mathbf{0.0861}$&$\mathbf{0.0967}$&$\mathbf{0.1083}$&$\mathbf{0.1205}$&$\mathbf{0.1335}$&$\mathbf{0.1472}$\\
				&(Std.Dev.)&$3.05\mbox{E-}{04}$&$2.56\mbox{E-}{04}$&$3.09\mbox{E-}{04}$&$2.78\mbox{E-}{04}$&$3.63\mbox{E-}{04}$&$3.43\mbox{E-}{04}$&$2.71\mbox{E-}{04}$&$3.30\mbox{E-}{04}$&$3.86\mbox{E-}{04}$&$4.58\mbox{E-}{04}$&$5.46\mbox{E-}{04}$\\
				\cline{2-13}&DBDP&$0.0422$&$0.0491$&$0.0568$&$0.0651$&$0.0739$&$0.0844$&$0.0952$&$0.1064$&$0.1186$&$0.1312$&$0.1449$\\
				&$95\%$ CI bounds&$\mathbf{[0.0421,0.0461]}$&$\mathbf{[0.0491,0.0529]}$&$\mathbf{[0.0568,0.0608]}$&$\mathbf{[0.0650,0.0700]}$&$\mathbf{[0.0738,0.0794]}$&$\mathbf{[0.0841,0.0894]}$&$\mathbf{[0.0952,0.1005]}$&$\mathbf{[0.1064,0.1119]}$&$\mathbf{[0.1184,0.1240]}$&$\mathbf{[0.1313,0.1372]}$&$\mathbf{[0.1448,0.1514]}$\\
				\cline{2-13}
				&$\widehat{\mathit{\Delta}}_{3,0}^{\Delta_t}$&$-0.2687$&$-0.2991$&$-0.3328$&$-0.3660$&$-0.4018$&$-0.4384$&$-0.4725$&$-0.5063$&$-0.5410$&$-0.5732$&$-0.6049$\\
				&(Std.Dev.)&$1.57\mbox{E-}{02}$&$1.80\mbox{E-}{02}$&$1.96\mbox{E-}{02}$&$2.46\mbox{E-}{02}$&$2.64\mbox{E-}{02}$&$3.28\mbox{E-}{02}$&$3.58\mbox{E-}{02}$&$4.17\mbox{E-}{02}$&$4.38\mbox{E-}{02}$&$5.01\mbox{E-}{02}$&$5.49\mbox{E-}{02}$\\
				\cline{1-13}
				\multirow{6}{*}{$20$}&$\widehat{Y}_0$&$\mathbf{0.0433}$&$\mathbf{0.0501}$&$\mathbf{0.0577}$&$\mathbf{0.0660}$&$\mathbf{0.0752}$&$\mathbf{0.0851}$&$\mathbf{0.0959}$&$\mathbf{0.1074}$&$\mathbf{0.1195}$&$\mathbf{0.1326}$&$\mathbf{0.1463}$\\
				&(Std.Dev.)&$3.48\mbox{E-}{04}$&$3.83\mbox{E-}{04}$&$4.97\mbox{E-}{04}$&$3.49\mbox{E-}{04}$&$3.77\mbox{E-}{04}$&$3.41\mbox{E-}{04}$&$3.50\mbox{E-}{04}$&$4.34\mbox{E-}{04}$&$4.15\mbox{E-}{04}$&$3.62\mbox{E-}{04}$&$3.82\mbox{E-}{04}$\\
				\cline{2-13}
				&DBDP&$0.0419$&$0.0490$&$0.0564$&$0.0652$&$0.0744$&$0.0844$&$0.0953$&$0.1061$&$0.1182$&$0.1307$&$0.1455$\\
				&$95\%$ CI bounds&$\mathbf{[0.0419,0.0450]}$&$\mathbf{[0.0489,0.0522]}$&$\mathbf{[0.0564,0.0598]}$&$\mathbf{[0.0652,0.0685]}$&$\mathbf{[0.0744,0.0775]}$&$\mathbf{[0.0843,0.0878]}$&$\mathbf{[0.0951,0.0990]}$&$\mathbf{[0.1058,0.1102]}$&$\mathbf{[0.1182,0.1217]}$&$\mathbf{[0.1305,0.1354]}$&$\mathbf{[0.1452,0.1491]}$\\
				\cline{2-13}
				&$\widehat{\mathit{\Delta}}_{3,0}^{\Delta_t}$&$-0.2624$&$-0.2945$&$-0.3286$&$-0.3624$&$-0.3975$&$-0.4342$&$-0.4691$&$-0.5046$&$-0.5407$&$-0.5753$&$-0.6031$\\
				&(Std. Dev.)&$1.50\mbox{E-}{02}$&$1.78\mbox{E-}{02}$&$1.91\mbox{E-}{02}$&$2.41\mbox{E-}{02}$&$2.83\mbox{E-}{02}$&$3.05\mbox{E-}{02}$&$3.52\mbox{E-}{02}$&$3.76\mbox{E-}{02}$&$4.21\mbox{E-}{02}$&$4.78\mbox{E-}{02}$&$5.43\mbox{E-}{02}$\\
				\hline
	\end{tabular}}}
\end{table}
The results of our RBSDE scheme within the SABSR model are quite promising; we observe that all of our approximations fall within the corresponding confidence intervals. Furthermore, our computational time (approximately $210$ seconds) is substantially lower than the training time of DPDB (approximately $637$ seconds). To elucidate the impact of stochastic correlation, Fig.\ref{fig:sec6_4} visually presents the disparities in price and delta-hedging between the results outlined in Table \ref{tab:sec6_7} and \ref{tab:sec6_8} for $N_T=400.$ Since the SCPs are parameterized to initialize and maintain a correlation value $-0.4,$ consistent with the constant correlation value utilized in Sec.\ref{sec:aocc}. The magnitude of differences in price and delta-hedging differences decreases with the increasing mean reversion speed $\kappa_{\rho},$ as displayed in Figure \ref{fig:sec6_4}. Consequently, the approximations with $\kappa_\rho=20$ (highest reversion speed) are expected to closely resemble those obtained with the constant correlation value, as depicted in Fig. \ref{fig:sec6_4}.
\begin{figure}[htbp!]
	\centering
	\begin{subfigure}[b]{0.4\textwidth}
		\includegraphics[width=\textwidth]{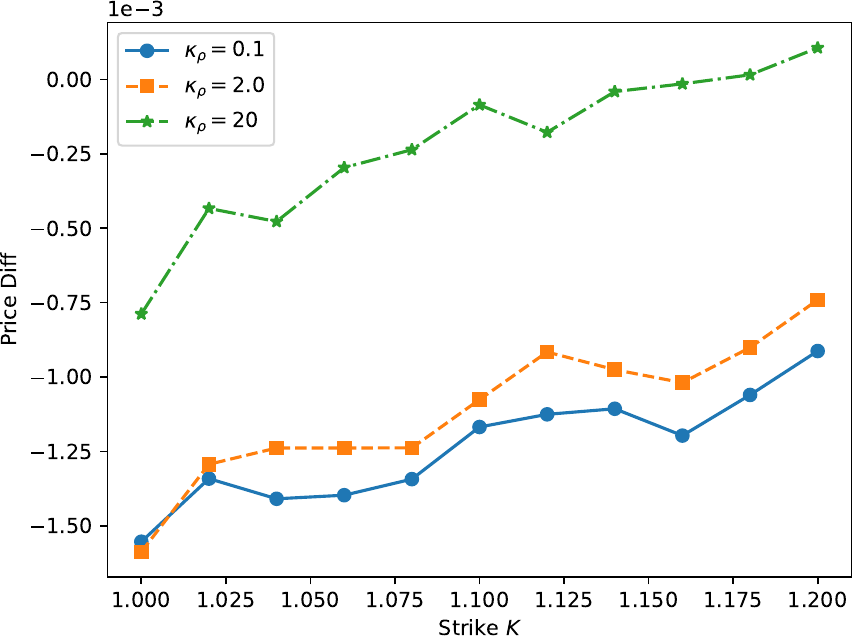}
		\subcaption{The price differences.}
	\end{subfigure}
	~
	\begin{subfigure}[b]{0.4\textwidth}
		\includegraphics[width=\textwidth]{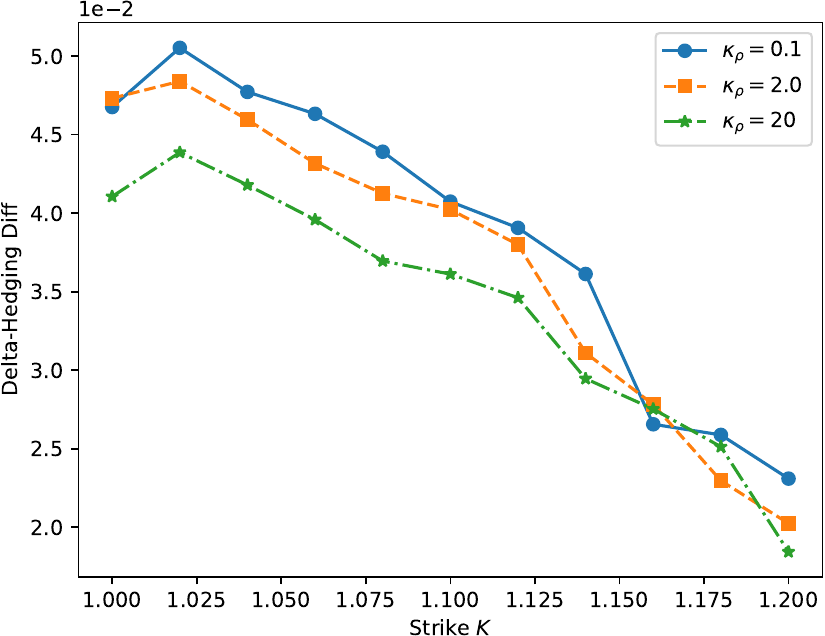}
		\subcaption{The delta-hedging differences.}
	\end{subfigure}
	\caption{The impact of the speed of mean reversion $\kappa_{\rho}$ on both the prices and delta-hedgings of American put options within the SABSR model.}\label{fig:sec6_4}
\end{figure}\\
The advantage of the RBSDE approach is that it allows us to capture the future evolution of the option price, which is particularly useful for tasks such as computing Value at Risk. To illustrate this, we display in Fig. \ref{fig:sec6_5a} the price evolution over the entire time domain by averaging the price data at each time point. For improved visualization, we set relatively large values for the spot and strike prices, namely $S_0 = 100$ and $K = 110$. All other parameters are chosen as in Table \ref{tab:sec6_8}, with $\kappa_\rho = 2$, except that we vary $\beta$. In addition, the distribution of the log-price data at $t = \tfrac{1}{2}$ is shown in Fig. \ref{fig:sec6_5b}.
\begin{figure}[htbp!]
	\centering
	\begin{subfigure}[b]{0.4\textwidth}
		\includegraphics[width=\textwidth]{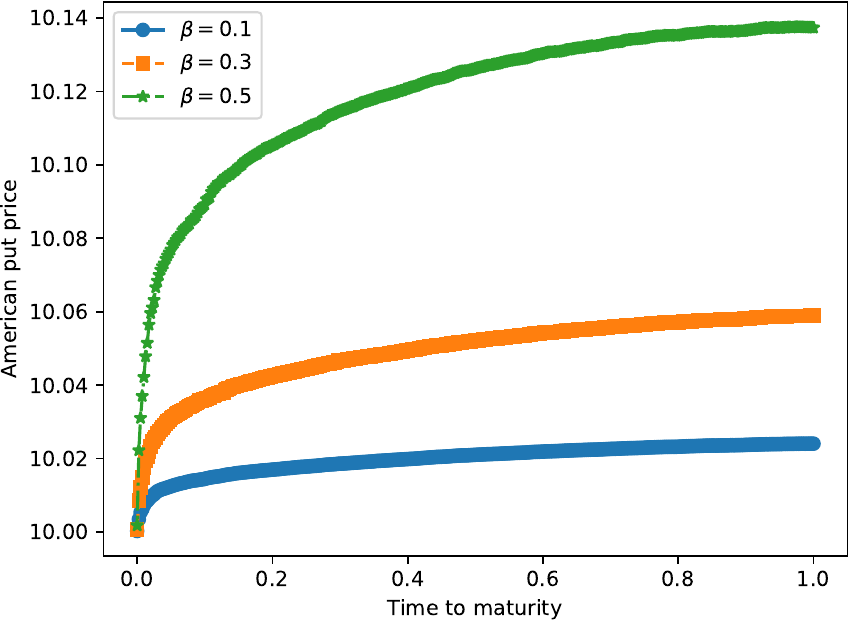}
		\subcaption{Price evolution.}\label{fig:sec6_5a}
	\end{subfigure}
	~
	\begin{subfigure}[b]{0.4\textwidth}
		\includegraphics[width=\textwidth]{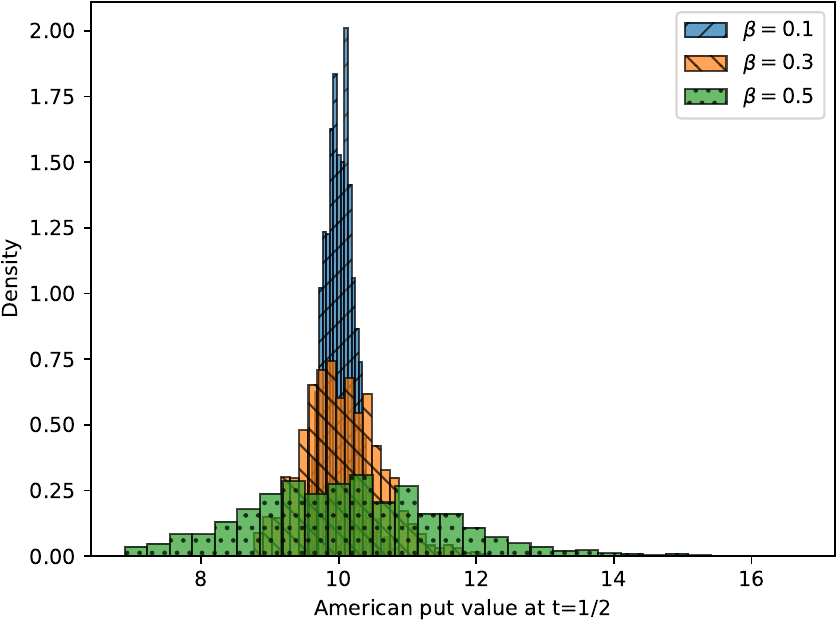}
		\subcaption{The log-price at $t=\frac{1}{2}.$}\label{fig:sec6_5b}
	\end{subfigure}
	\caption{Future evolution of the American put option price under the SABSR model for varying $\beta.$ $S_0=100, K=110,$ and all remaining parameter values follow those reported in Table \ref{tab:sec6_8}, with $\kappa_\rho = 2.$}
\end{figure}\\
\subsection{Numerical Example: Heston SLV Model with Stochastic Correlation}
The SCLV model examined in this section, as detailed in Sec.\ref{sec:extensionHestinSLV}, is an extension of the Heston SLV (HSLV) model that incorporates stochastic correlation, referred to as the HSCLV model. Different to the SABSR model, both the local volatility and underlying asset processes have nonzero drift terms in the HSCLV model.

\subsubsection{American options under the HSLV model}
We revisit the valuation of American put options under the HSLV model using our RBSDE-based approach. Once again, we compare our approximations with those obtained using the CTMC method, as presented in Table \ref{tab:sec6_11}.
\begin{table}[h!]
	\centering
	\caption{American put options. Option parameters: $S_0=10, T=1, \beta=0.7, \nu_0=0.04, \kappa_\nu=4, \mu_\nu=0.035, \sigma_\nu=0.15, r=0.03.$ Constant correlation is $\rho=-0.75.$ RBSDE parameters: $M=200000, \mu_s=0.03.$}\label{tab:sec6_11}
	\resizebox{\columnwidth}{!}{
		\begin{tabular}{c|ccccccccc}
			\hline
			&&\multicolumn{7}{c}{$K$}\\
			\cline{3-9}
			$N_T$&&$7.0$&$8.0$&$9.0$&$10.0$&$11.0$&$12.0$&$13.0$\\
			\hline
			\multirow{7}{*}{$10$}&CTMC&$\mathbf{0.0002}$&$\mathbf{0.0044}$&$\mathbf{0.0473}$&$\mathbf{0.2813}$&$\mathbf{0.9999}$&$\mathbf{1.9999}$&$\mathbf{2.9999}$\\
			\cline{2-9}
			&$\widehat{Y}_0$&$0.0002$&$0.0036$&$0.0444$&$0.2721$&$0.9999$&$1.9999$&$2.9999$\\
			&(Std.Dev.)&$(2.24\mbox{E-}{05})$&$(1.13\mbox{E-}{04})$&$(4.86\mbox{E-}{04})$&$(1.01\mbox{E-}{03})$&$(0.00\mbox{E-}{00})$&$(4.68\mbox{E-}{16})$&$(4.68\mbox{E-}{16})$\\
			\cline{2-9}
			&$\widehat{\mathit{\Delta}}_{3,0}^{\Delta_t}$&$-0.0001$&$-0.0073$&$-0.0859$&$-0.3939$&$-0.9445$&$-0.9852$&$-0.9832$\\
			&(Std. Dev.)&$(1.66\mbox{E-}{04})$&$(2.42\mbox{E-}{04})$&$(3.60\mbox{E-}{03})$&$(5.25\mbox{E-}{03})$&$(1.61\mbox{E-}{02})$&$(2.93\mbox{E-}{02})$&$(4.36\mbox{E-}{02})$\\
			\cline{2-9}
			& Average running time &\multicolumn{7}{c}{$6$ seconds}\\
			\hline
			\multirow{5}{*}{$20$}&$\widehat{Y}_0$&$0.0002$&$0.0041$&$0.0462$&$0.2776$&$0.9999$&$1.9999$&$2.9999$\\
			&(Std.Dev.)&$(1.77\mbox{E-}{05})$&$(7.63\mbox{E-}{05})$&$(2.88\mbox{E-}{04})$&$(1.05\mbox{E-}{03})$&$(0.00\mbox{E-}{00})$&$(4.68\mbox{E-}{16})$&$(4.68\mbox{E-}{16})$\\
			\cline{2-9}
			&$\widehat{\mathit{\Delta}}_{3,0}^{\Delta_t}$&$-0.0001$&$-0.0081$&$-0.0886$&$-0.3893$&$-0.9666$&$-0.9854$&$-0.9808$\\
			&(Std. Dev.)&$(7.27\mbox{E-}{05})$&$(1.71\mbox{E-}{04})$&$(3.27\mbox{E-}{03})$&$(6.60\mbox{E-}{03})$&$(1.36\mbox{E-}{02})$&$(2.78\mbox{E-}{02})$&$(4.19\mbox{E-}{02})$\\
			\cline{2-9}
			& Average running time &\multicolumn{7}{c}{$11$ seconds}\\
			\hline
			\multirow{4}{*}{$30$}&$\widehat{Y}_0$&$\mathbf{0.0002}$&$\mathbf{0.0046}$&$\mathbf{0.0474}$&$\mathbf{0.2811}$&$\mathbf{0.9999}$&$\mathbf{1.9999}$&$\mathbf{2.9999}$\\
			&(Std.Dev.)&$(4.41\mbox{E-}{05})$&$(1.07\mbox{E-}{04})$&$(6.13\mbox{E-}{04})$&$(1.32\mbox{E-}{03})$&$(0.00\mbox{E-}{00})$&$(4.68\mbox{E-}{16})$&$(4.68\mbox{E-}{16})$\\
			\cline{2-9}
			&$\widehat{\mathit{\Delta}}_{3,0}^{\Delta_t}$&$-0.0000$&$-0.0091$&$-0.0897$&$-0.3851$&$-0.9725$&$-0.9794$&$-0.9702$\\
			&(Std. Dev.)&$(3.99\mbox{E-}{05})$&$(1.93\mbox{E-}{04})$&$(2.96\mbox{E-}{03})$&$(5.75\mbox{E-}{03})$&$(1.81\mbox{E-}{02})$&$(3.42\mbox{E-}{02})$&$(5.14\mbox{E-}{02})$\\
			\cline{2-9}
			& Average running time &\multicolumn{7}{c}{$16$ seconds}\\
			\hline
	\end{tabular}}
\end{table}
The results further demonstrate the high accuracy of our RBSDE scheme for pricing within the HSLV framework.
\subsubsection{American options under the HSCLV model}
Finally, we examine the evaluation of American options using the HSCLV model. For the SCP, we set $\rho_0=-0.75,$ mean reversion speed $\kappa_\rho=2.0,$ volatility $\sigma_\rho=0.1,$ and correlation between the SCP and the volatility process $\rho_{\rho\nu}=0.$ We then vary the long-term mean $\mu_\rho.$ In other words, the SCP is initialized with the same values but is allowed to converge to different long-term mean values while maintaining the same mean-reversion speed. The results are reported in Table \ref{tab:sec6_12}, along with the corresponding DBDP approximations. The same conclusion can be drawn as in the SABSR example. In particular, for the case with $N_T=30,$ our computational time is approximately $35$ seconds, compared with $639$ seconds required for training DPDB scheme.
\begin{table}[h!]
	\centering
	\caption{American put options. Option parameters: $S_0=10, T=1, \beta=0.7, \nu_0=0.04, \kappa_\nu=4, \mu_\nu=0.035, \sigma_\nu=0.15, r=0.03.$ SCP: $\rho_0=-0.75, \kappa_\rho=2.0, \sigma_\rho=0.1, \rho_{\rho \nu}=0.$ RBSDE parameters: $M=200000, N_T=30, \mu_s=0.03.$ The average running time of our scheme is approximately $35$ seconds for the given parameter setting, and the training time of DPDB scheme is about $639$ seconds.}\label{tab:sec6_12}
	\resizebox{\columnwidth}{!}{
		\begin{tabular}{c|ccccccccc}
			\hline
			&&\multicolumn{7}{c}{$K$}\\
			\cline{3-9}
			$\mu_{\rho}$&&$7.0$&$8.0$&$9.0$&$10.0$&$11.0$&$12.0$&$13.0$\\
			\hline
			\multirow{6}{*}{$-0.2$}&$\widehat{Y}_0$&$0.0001$&$0.0040$&$0.0452$&$0.2794$&$0.9999$&$1.9999$&$2.9999$\\
			&(Std.Dev.)&$(1.78\mbox{E-}{05})$&$(1.53\mbox{E-}{04})$&$(5.59\mbox{E-}{04})$&$(1.09\mbox{E-}{03})$&$(0.00\mbox{E-}{00})$&$(4.68\mbox{E-}{16})$&$(4.68\mbox{E-}{16})$\\
			&DPDB&$0.0001$&$0.0035$&$0.0440$&$0.2776$&$0.9937$&$1.9928$&$2.9921$\\
			&$95\%$ CI bounds&$[0.0001,0.0011]$&$[0.0034,0.0062]$&$[0.0438,0.0475]$&$[0.2769,0.2814]$&$[0.9937,1.0002]$&$[1.9926,2.0000]$&$[2.9920,3.0001]$\\
			\cline{2-9}
			&$\widehat{\mathit{\Delta}}_{3,0}^{\Delta_t}$&$-0.0002$&$-0.0093$&$-0.0893$&$-0.3946$&$-0.9817$&$-0.9995$&$-0.9996$\\
			&(Std. Dev.)&$(5.69\mbox{E-}{05})$&$(4.48\mbox{E-}{04})$&$(2.89\mbox{E-}{03})$&$(5.97\mbox{E-}{03})$&$(2.06\mbox{E-}{02})$&$(3.76\mbox{E-}{02})$&$(5.52\mbox{E-}{02})$\\
			\hline
			\multirow{6}{*}{$0$}&$\widehat{Y}_0$&$0.0001$&$0.0037$&$0.0440$&$0.2786$&$0.9999$&$1.9999$&$2.9999$\\
			&(Std.Dev.)&$(1.94\mbox{E-}{05})$&$(1.25\mbox{E-}{04})$&$(4.10\mbox{E-}{04})$&$(1.04\mbox{E-}{03})$&$(0.00\mbox{E-}{00})$&$(4.68\mbox{E-}{16})$&$(4.68\mbox{E-}{16})$\\
			&DPDB&$0.0001$&$0.0031$&$0.0427$&$0.2767$&$0.9939$&$1.9929$&$2.9919$\\
			&$95\%$ CI bounds&$[0.0001,0.0010]$&$[0.0031,0.0058]$&$[0.0425,0.0462]$&$[0.2763,0.2802]$&$[0.9937,1.0002]$&$[1.9928,2.0000]$&$[2.9917,3.0002]$\\
			\cline{2-9}
			&$\widehat{\mathit{\Delta}}_{3,0}^{\Delta_t}$&$-0.0002$&$-0.0089$&$-0.0888$&$-0.3977$&$-0.9812$&$-0.9999$&$-0.9998$\\
			&(Std. Dev.)&$(6.67\mbox{E-}{05})$&$(3.43\mbox{E-}{04})$&$(2.59\mbox{E-}{03})$&$(6.70\mbox{E-}{03})$&$(2.01\mbox{E-}{02})$&$(3.74\mbox{E-}{02})$&$(5.49\mbox{E-}{02})$\\
			\hline
			\multirow{6}{*}{$0.75$}&$\widehat{Y}_0$&$0.0000$&$0.0028$&$0.0395$&$0.2755$&$0.9999$&$1.9999$&$2.9999$\\
			&(Std.Dev.)&$(9.36\mbox{E-}{06})$&$(1.23\mbox{E-}{04})$&$(2.91\mbox{E-}{04})$&$(8.50\mbox{E-}{04})$&$(0.00\mbox{E-}{00})$&$(4.68\mbox{E-}{16})$&$(4.68\mbox{E-}{16})$\\
			&DPDB&$0.0000$&$0.0021$&$0.0382$&$0.2740$&$0.9939$&$1.9928$&$2.9921$\\
			&$95\%$ CI bounds&$[0.0000,0.0008]$&$[0.0021,0.0044]$&$[0.0380,0.0419]$&$[0.2734,0.2763]$&$[0.9938,1.0002]$&$[1.9926,2.0000]$&$[2.9920,3.0002]$\\
			\cline{2-9}
			&$\widehat{\mathit{\Delta}}_{3,0}^{\Delta_t}$&$-0.0000$&$-0.0074$&$-0.0838$&$-0.4014$&$-0.9795$&$-0.9998$&$-0.9996$\\
			&(Std. Dev.)&$(3.12\mbox{E-}{05})$&$(3.02\mbox{E-}{04})$&$(1.67\mbox{E-}{03})$&$(6.55\mbox{E-}{03})$&$(2.07\mbox{E-}{02})$&$(3.73\mbox{E-}{02})$&$(5.49\mbox{E-}{02})$\\
			\hline
	\end{tabular}}
\end{table}
It is noteworthy that the impact of stochastic correlation almost disappears for deep in-the-money options with the given parameter values. This phenomenon is illustrated in Fig. \ref{fig:sec6_6}.
\begin{figure}[h!]
	\centering
	\begin{subfigure}[b]{0.4\textwidth}
		\includegraphics[width=\textwidth]{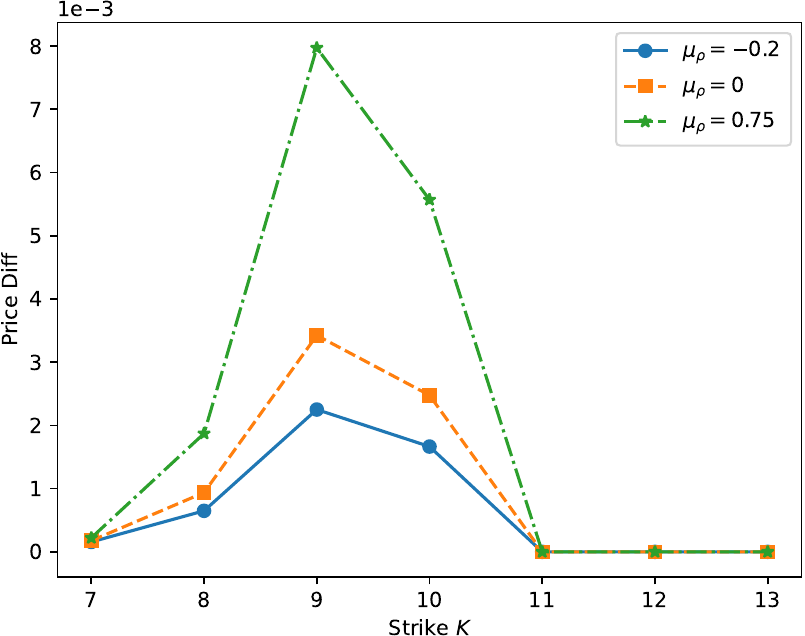}
		\subcaption{The price differences.}
	\end{subfigure}
	~
	\begin{subfigure}[b]{0.4\textwidth}
		\includegraphics[width=\textwidth]{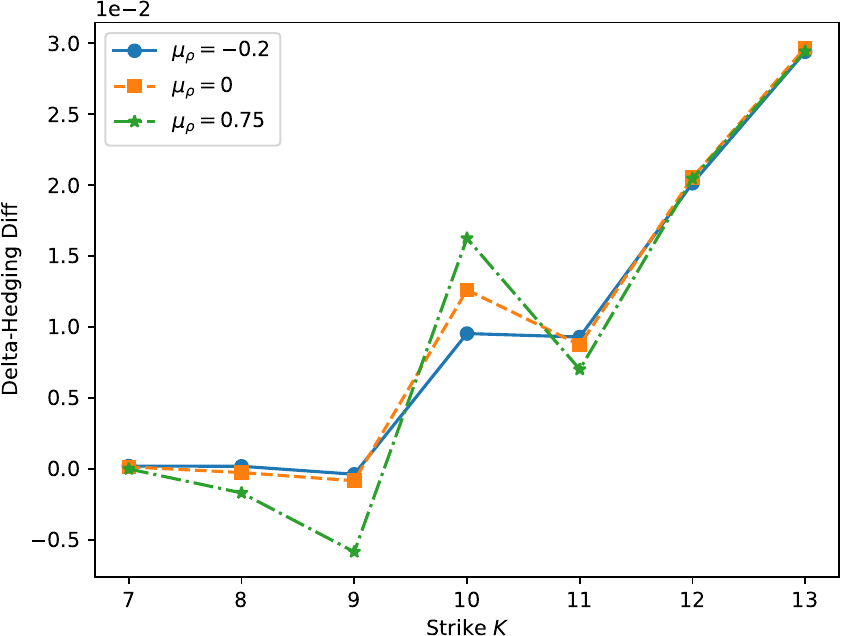}
		\subcaption{The delta-hedging differences.}
	\end{subfigure}
	\caption{The impact of the long-term mean of SCP $\mu_{\rho}$ on both the prices and delta-hedgings of American put options within the HSCLV model.}\label{fig:sec6_6}
\end{figure}
We may conclude that stochastic correlation plays a minor role when other parameters have dominant values in determining the option value. For a better visualization in Fig. \ref{fig:sec6_7}, we set relatively large values for the spot and strike prices, namely $S_0 = 100$ and $K = 105.$ All other parameters are chosen as in Table \ref{tab:sec6_12}, with $\mu_{\rho} = -0.2$, except that we vary $\beta$. In addition, the distribution of the log-price data at $t = \frac{1}{2}$ is reported in Fig. \ref{fig:sec6_6b}.
\begin{figure}[htbp!]
	\centering
	\begin{subfigure}[b]{0.4\textwidth}
		\includegraphics[width=\textwidth]{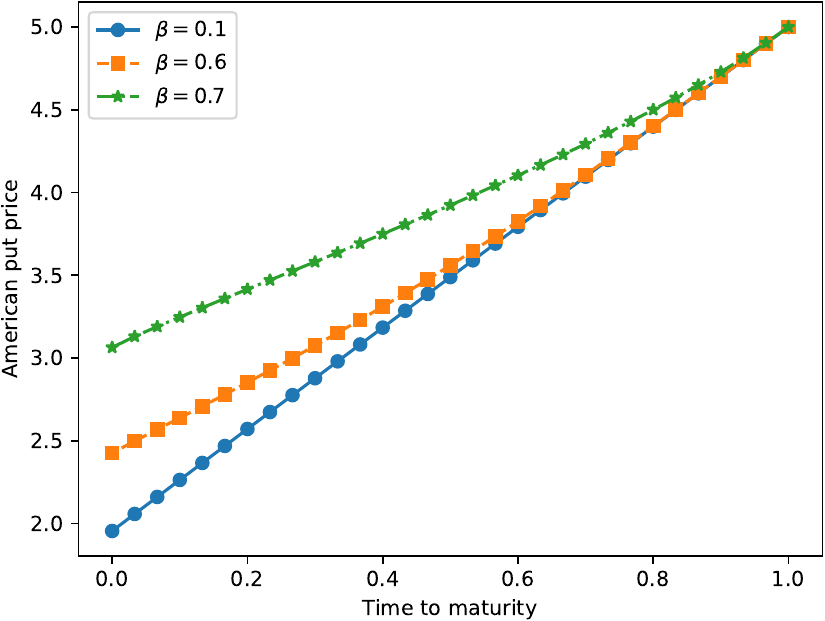}
		\subcaption{Price evolution.}\label{fig:sec6_6a}
	\end{subfigure}
	~
	\begin{subfigure}[b]{0.4\textwidth}
		\includegraphics[width=\textwidth]{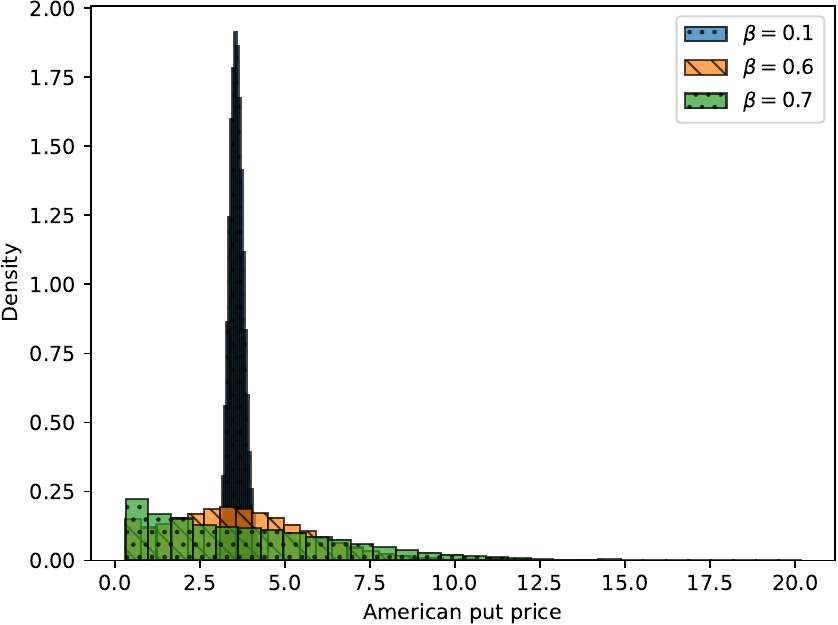}
		\subcaption{The log price at $t=\frac{1}{2}.$}\label{fig:sec6_6b}
	\end{subfigure}
	\caption{Future evolution of American put option price under the HSCLV model for varying $\beta.$ $S_0=100,$ $K=105,$ and all remaining parameter values follow those reported in Table \ref{tab:sec6_12}, with $\mu_{\rho} = -0.2.$}\label{fig:sec6_7}
\end{figure}\\

\section{Conclusion}\label{sec:conclusion}
One method to hedge correlation risk is to incorporate stochastic correlation into pricing and hedging strategies. This work studies the pricing of American options under SLV models, extended by introducing stochastic correlation driven by a stochastic process, referred to as the SCLV models. The SLV models already pose significant challenges for standard pricing techniques, and these challenges are further compounded in the SCLV framework.

To address these challenges, we propose an RBSDE–based approach for pricing and hedging American options under stochastic correlation, volatility, and spot price dynamics. A key advantage of the RBSDE methodology is that it also enables computation of the future evolution of the option price. To solve the resulting RBSDEs, we develop a numerical algorithm based on XGBoost regression. The SABSR and HSCLV models are analyzed in detail for the pricing and hedging of American options. Comparisons with state-of-the-art deep learning methods demonstrate the effectiveness of the RBSDE approach, the high accuracy of the regression-based algorithm, and the significant impact of stochastic correlation on both pricing and delta-hedging.

Future research directions include addressing the challenges posed by ill-conditioned systems of equations, with the goal of obtaining robust numerical approximations for hedging under stochastic volatility and correlation. Moreover, there is a growing need for faster numerical methods for solving RBSDEs, particularly to enable efficient calibration of SCLV models.

\section*{Acknowledgement(s)}

We gratefully acknowledge the authors of \cite{Hure2020} and \cite{Yang2024b} for making their code available. \\No competing interest was reported by the authors(s).

\section*{Funding}
The work was funded by the Deutsche Forschungsgemeinschaft (DFG, German Research Foundation)- Project No. [581639701].

\end{document}